\documentclass[sigconf]{acmart}

\usepackage{bbm}
\usepackage{mathrsfs}
\usepackage{graphicx}
\usepackage{amsfonts}
\usepackage{flushend}
\usepackage{amsmath}
\usepackage{graphicx}
\usepackage{subfigure}
\usepackage{epsfig}
\usepackage{latexsym}
\usepackage{amsfonts}
\usepackage{amssymb}
\usepackage{paralist}
\usepackage{xspace}
\usepackage{mathrsfs}
\usepackage{amssymb}
\usepackage{color}
\usepackage{flushend}
\usepackage{paralist}
\usepackage{multirow}
\usepackage{dsfont}
\usepackage{multicol}
\usepackage{etoolbox}
\usepackage{courier}
\usepackage{booktabs}
\usepackage{pifont}
\usepackage[figurename=Fig.]{caption}

\def\ie{i.e.,\xspace}
\def\etal{et al.\xspace}
\def\etc{etc.}
\def\eg{e.g.,\xspace}

\newcommand{\oursystem}{Tash\xspace}

\newtheorem{thm}{Theorem~}

\newtheorem{lem}{Lemma~}
\newtheorem{defn}{Definition~}
\newtheorem{problem}{Problem~}

\renewcommand\footnotetextcopyrightpermission[1]{} % removes footnote with conference information in first column
\makeatletter
\renewcommand\@formatdoi[1]{\ignorespaces}
\makeatother

\definecolor{dark-red}{HTML}{990033}

\graphicspath{{figs/}}

\usepackage{etoolbox}% http://ctan.org/pkg/etoolbox
\makeatletter
\patchcmd{\maketitle}{\@copyrightspace}{}{}{}
\makeatother

\begin{document}
\title{Analog On-Tag Hashing: Towards Selective Reading as Hash Primitives in Gen2 RFID Systems}
\subtitle{(Technical Report)}
\author{Lei Yang$^{\ast}$,  Qiongzheng Lin$^{\ast}$, Chunhui Duan$^{\ast}$$^{\dagger}$,  Zhenlin An$^{\ast}$ } 
\affiliation{%
      \institution{$^\ast$Department of Computing, The Hong Kong Polytechnic University}
	\institution{$^\dagger$ School of Software, Tsinghua University}
 }
\email{{young, lin, hui, an}@tagsys.org}

\begin{abstract}

Deployment of  billions of Commercial off-the-shelf (COTS) RFID tags has drawn much of the attention of the research community because of the performance gaps of current systems. In particular, hash-enabled protocol (HEP)  is one of the most thoroughly studied topics in the past decade.  HEPs are designed for a wide spectrum of notable applications (\eg missing detection) without need to collect all tags. HEPs assume that each tag contains a \emph{hash function}, such that a tag can select a \emph{random} but \emph{predicable} time slot to reply with a \emph{one-bit} presence signal that shows its existence. However, the hash function has never been implemented in COTS tags in reality, which makes HEPs a 10-year untouchable mirage.  This work designs and implements a group of  analog on-tag hash primitives (called \emph{Tash}) for COTS Gen2-compatible RFID systems,  which moves prior HEPs forward from theory to practice.  In particular, we design three types of hash primitives, namely, \emph{tash function}, \emph{tash table function} and \emph{tash operator}. All of these hash primitives are implemented through \emph{selective reading}, which is a fundamental and mandatory functionality specified in Gen2 protocol, without any hardware modification and fabrication. We further apply our hash primitives in two typical HEP applications (\ie cardinality estimation and missing detection) to show the feasibility and effectiveness of Tash. Results from our prototype, which is  composed of one ImpinJ reader and $3,000$ Alien tags, demonstrate that the new design lowers $60\%$ of the communication overhead in the air. The tash operator can additionally introduce an overhead drop of $29.7\%$.

\end{abstract}




\maketitle

\section{Introduction}
\label{section:introduction}

RFID systems are increasingly used in everyday scenarios, which range from object tracking, indoor localization \cite{yang2014tagoram}, vibration sensing \cite{yang2016making}, to medical-patient management, because of the extremely low cost of commercial RFID tags (\eg as low as 5 cents per tag). 
%In fact, their costs are considerably lower than those of products to which they are attached.
%For example, a typical application involves tagging all items in a store or warehouse, such that purchases or records can be performed automatically when pushing a shopping cart or cargo through the checkout line\cite{barat2013rfid}.
%%%Reserved for trans
Recent reports show that many industries like healthcare and retailing are moving towards deploying RFID systems for object tracking, asset monitoring, and emerging Internet of Things \cite{industry-report}.

\subsection{The State-of-the-Art}
The Electronic Product Code global is an organization established  to  accomplish the worldwide adoption and standardization of EPC technology. It published the Gen2 air protocol \cite{gen2} for RFID system in 2004. A Gen2 RFID system consists of a reader and many passive tags. The passive tags without batteries are  powered up purely by harvesting radio signals from readers.  This protocol has become the mainstream specification globally, and has been adopted as a major part of the ISO/IEC 18000-6 standard.

Embedding Gen2 tags into everyday objects to construct ubiquitous networks has been a long-standing vision. However, a major problem that challenges this vision is that the Gen2 RFID system is not efficient \cite{wang2012efficient}. First, the RFID system utilizes  simple modulations (e.g., ON-OFF keying or BPSK) due to the lack of traditional transceiver \cite{dobkin2012rf}, which prevents tags from leveraging a suitable channel to transmit more bits per symbol and increase the bandwidth efficiency.  Second, tags cannot hear the transmissions of other tags. They merely reply on the reader to schedule their medium access with the Framed Slotted ALOHA protocol, which results in many  empty and collided slots. This condition also retards the inventory process. These two limitations force a reader to go through a  long inventory phase when it collects all the tags in the scene.

\subsection{Ten-Year Mirage of HEP}

Motivated by the aforementioned performance gaps, the research community opened a new focus on HEP design approximately 10 years ago.  The key idea that underlies HEPs is that each tag selects a time slot according to the hash value of its \texttt{EPC} and a random seed. It then replies a one-bit presence signal rather than the entire \texttt{EPC} number in the selected slot. HEPs treat all tags as if they were a virtual sender, which outputs a gimped hash table  (\ie a \emph{presence bitmap}) when responding to a challenge (\ie a random seed). Most importantly, HEPs assume the backend server and every tag share a \emph{hash function}, and the resulting bitmap is random but predicable when the \texttt{EPC}s and seeds are known.

Fig.~\ref{fig:hep} shows a toy example with $n=8$ tags, each of which contains a unique \texttt{EPC} number presented in binary format (\eg $101010_2$), to illustrate the HEP concept. The reader divides the time into $d$ time slots (\eg $d=8$.) and challenges these tags with the random seed $r$. Each tag selects the $(h_d(\text{\texttt{EPC}},r))^{th}$ time slot to reply the one-bit signal, where $h(\cdot)$ is a common hash function (\eg MD5, SHA-1) and $h_d(\cdot) = h(\cdot)\bmod d$. The reader can recognize two possible results for each time slot, namely, \emph{empty} and \emph{non-empty}\footnote{Some work assume the reader can recognize the signal collision, obtaining three results: empty, single and collision. }. The reader abstracts the reply results into a bitmap (\ie $B=[0,1,1,0,1,1,1,0]$), where each element  contains two possible values, that is, $0$ and $1$, that corresponds to empty and non-empty slots, respectively.  The upper layer then utilizes this returned bitmap to explore many notable  applications. We show the following two typical applications as examples to drive the key point:

$\bullet$ \emph{Cardinality estimation.} Estimating the size of a given tag population is required in many applications, such as privacy sensitive systems and warehouse monitoring. Kodialam \etal\cite{kodialam2006fast} presented a pioneer estimator. Given that tags select the time slots uniformly because of hashing, the expected number of `$0$'s equals $n_0=d (1-1/d)^n \approx de^{-n/d}$. Counting $n_0$ in an instance yields a ``zero estimator'', \ie $\widehat{n}\approx -d\ln(n_0/d)$. For example, $\widehat{n}=-8\times\ln(3/8)=7.8$ in our toy example.

$\bullet$ \emph{Missing detection.} Consider a major warehouse that stores thousands of apparel, shoes, pallets, and cases. How can a staff \emph{immediately} determine if anything is missing?  Sheng and Li \cite{tan2008monitor} conducted the early study on the fast detection of missing-tag events by using the presence bitmap. They assumed all \texttt{EPC}s were known in a closed system. Given that hash results are predicable, the system can generate an \emph{intact} bitmap at the backend. We can identify the missing tags in a probabilistic approach by comparing the intact and instanced bitmaps.  For example, if the second entry equals $0$ (which is supposed to be $1$), the the tag $101010_2$ must be missing in our toy example.

\begin{figure}[t!]
  \centering
  \includegraphics[width=0.8\linewidth]{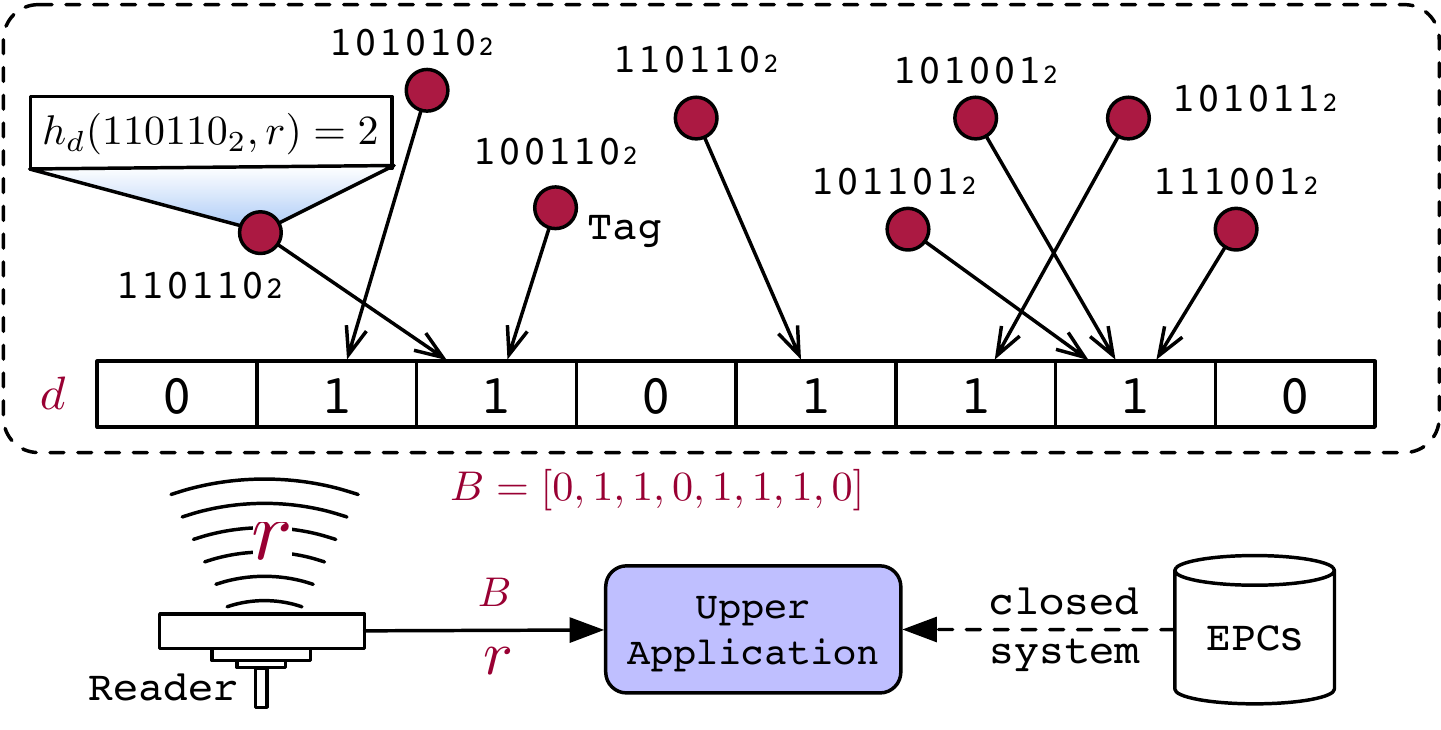}
   \caption{Hash enabled protocol illustration. \textnormal{In the figure, $8$ tags emit one-bit signals in the $h_d(\texttt{EPC},r)^{th}$ time slots respectively, which are challenged by the random seed $r$ and the frame length $d$. Finally, the reader abstracts tags' responses as a presence bitmap.}}
  \label{fig:hep}
    \vspace{-0.2cm}
\end{figure}

HEPs are advantageous in terms of speed and privacy. HEPs are faster than all prior per-tag reading schemes for two reasons. First, collecting all the EPCs of the tags is time consuming  because of the aforementioned low-rate modulation, whereas one-bit presence signals of HEPs save approximately  $96\times$ of the time (\ie the \texttt{EPC} length equals 96 bits in theory \footnote{Actual case in practice would be less than this estimate due to other extra jobs, such as setup time, query time, \etc}). Second, collisions  are considered as one of the major reasons that drag down the reading. On the contrary, HEPs tolerate and consider collisions as informative.  When privacy issues are considered, the tag's identification may be unacceptable in certain instances.  HEPs allow tags to send out non-identifiable information (\ie one-bit signals).

HEPs are very promising. However, after 10 years of enthusiastic discussion about the opportunities that HEPs provide, the reality is beginning to settle: the functionality of hashing (\ie hash function and hash table function) has never been implemented  in any Gen2 RFID tags and considered by any RFID standard. No hint shows that this function will be widely accepted in the near future.

\subsection{Why Not Support Hashing?}

A large number of recent work have attempted to supplement hash functionality to RFID tags, which can be categorized into three groups. First group, like \cite{feldhofer2006case,poschmann2007new}, modifies the common hash functions to accommodate resource-constrained RFID tags. The second group \cite{yoshida2007mame,bogdanov2007present,rolfes2008ultra,bogdanov2007present,poschmann2007new,lim2005mcrypton,yusecurity,hong2006hight,good2007hardware} designs new lightweight and efficient hash functions dedicatedly for RFID tags. The third group seeks new design of RFID tags like WISP\cite{philipose2005battery} and Moo \cite{zhang2011moo}, which gives tags more powerful computing capabilities (\eg hashing \cite{pendl2012elliptic}). Unfortunately, as far as we know, none of these work has been really applied in COTS RFID systems yet.

Why is the hash function unfavored? A term called as \emph{Gate Equivalent} (GE) is widely used to evaluate a hardware design with respect to its efficiency and availability. One GE is esquivalient to the area which is required by the two-input NAND gate with the lowest deriving strength of the corresponding technology.  A glance at Table.~\ref{tab:hash-comparision} shows the available designs of hash functions for RFID tags require a significant number of GEs, which are completely unaffordable by current COTS tags. For example, the most compact hash functions requires thousands of GEs  (\eg $1,075$ GEs for PRESENT-80), which  incur extremely high energy consumption and manufacture cost. Thus, relatively few RFID-oriented protocols that appeal to a hash function can be utilized. RFID was expected to be one of the most competitive automatic identification technologies  due to its many attractive advantages (\eg simultaneous reading, NLOS, \etc) compared with others (\eg barcode). However, this progress has been hindered for many years by the final obstacle that the industry is attempting to overcome (\ie the price). The industry is extremely sensitive to the cost being doubled or tripled  by the hash, although HEPs actually introduce significant  outperformance.

\subsection{Our  Contributions}

This work designs a group of  hash primitives, \emph{Tash},  which takes advantage of existing fundamental function of \emph{selective reading} specified in Gen2 protocol, \emph{without} any hardware modification and fabrication. Our design and implementation both strictly follow the Gen2 specification, so it can work in any Gen2-Compatible RFID system. These mimic (or analog) hash primitives act as we embedded real hash  circuits on tags\footnote{This work does not target at designing any analog circuit on readers or tags, but offers a mimic hash function acting as we  embed a hash circuit on each tag.}, while we actually implement them in application layer. Specifically, we design the following three kinds of hash primitives to revive prior HEPs:

$\bullet$ We design a hash function (aka tash function)  over existing COTS Gen2 tags. The hash function outputs a hash value associated with the \texttt{EPC} of the tag and a random seed, as HEPs require.

$\bullet$ We design a hash table function (aka tash table function) over all tags in the scene. It can produce a  hash table (aka tash table), which is more informative than a bitmap, over the all tags in the scene. In particular, each entry indicates the exact number of tags hashed into this entry.

$\bullet$ Major prior HEPs require multiple acquisitions of bitmaps to meet an acceptable confidence. We design three tash operators (\ie tash AND, OR and XOR) to perform entry-wise set operations over multiple tash tables on tag in the physical layer, which offers a one-stop acquisition solution.

\textbf{Summary.} It  has been considered that HEPs are hardly applied in practice because of the `impossible mission'  of implementing hash function on COTS Gen2 tags \cite{bogdanov2008hash}. In this work, our main contribution lies in the practicality and usability, that is, enabling billions of deployed tags to benefit performance boost from prior well-studied HEPs, with our hash primitives.  To the best of our knowledge,  this is the first work to implement the hash functionality over COTS Gen2 tags.  Second, we provide an implementation of \oursystem and show its feasibility  and efficiency in two typical usage scenarios.  Third, we investigate several leading RFID products in  market  including $18$ types of tags and $10$ types of readers, in terms of their compatibility with Gen2, and conduct an extensive evaluation on our prototype with COTS devices.

%\noindent\textbf{Summary of Results.} We build a software framework using LLRP protocol to implement all hashing primitives aforementioned. We also demonstrate the effectiveness of the framework over the Commercial Off-The-Shelf (COTS) ImpinJ R420 reader and Alien tags.  Our experiments lead to the following finds:

%====reserved for extension====
%\subsection{Organization}
%The reminder of this paper is organized as follows. $\S$\ref{section:motivation} list several key usage scenarios of HEPs motivating our design. $\S$\ref{section:overview} overviews the system and formally defines the hash primitives. We present the tash's design, usage study and implementation in $\S$\ref{section:motivation}, $\S$\ref{section:usage-study} and $\S$\ref{section:implementation} respectively. $\S$\ref{section:evaluation} show the benchmark and case study. Finally, we introduce the related work and give conclusions in $\S$\ref{section:related-work} and $\S$\ref{section:conclusion}, respectively.

\section{Related Work}
\label{section:motivation}

We review the related work from two aspects: the designs of hash functions and hash enabled protocols.

\textbf{Design of hash function.} Feldhofer and Rechberger \cite{feldhofer2006case} firstly point that current common hash functions (\eg \texttt{MD5}, \texttt{SHA-1}, \etc), are not hardware friendly and unsuitable at all for RFID tags, which have very constrained computing ability.   Such difficulty has spurred considerable research \cite{feldhofer2006case,yoshida2007mame,bogdanov2007present,rolfes2008ultra,bogdanov2007present,poschmann2007new,lim2005mcrypton,yusecurity,hong2006hight,poschmann2007new,good2007hardware}. We sketch the primary designs and their features in Table.~\ref{tab:hash-comparision}. For example, Bogdanov \etal \cite{bogdanov2007present} propose a hardware-optimized block cipher, PRESENT, designed with area and power constraints.  The follow-up work \cite{rolfes2008ultra} presents three different architectures of PRESENT and highlights their availability for both active and passive smart devices. Their implementations reduce the number of GEs to $1,000$ around. Another follow-up work \cite{bogdanov2008hash} extends the design of PRESENT and gives 8 variants to fulfill different requirements, \eg DM-PRESENT-80, DM-PRESENT-128, H-PRESENT-128, \etc\  The work \cite{poschmann2007new} suggests to choose DES as hash function for RFID tags due to relatively low complexity, and presents a variant of DES, called as\ie DESXL. Lim and Korkishko \cite{lim2005mcrypton} present a $64$-bit hash function with three key size options ($64$ bits, $96$ bits and $128$ bits),  which requires  about $3,500$ and $4,100$ GEs.  In summary, despite these optimized designs, majority are still presented in theory and  none of them are available for the COTS RFID tags. On contrary, our work explores hash function from another different aspect, that is, leveraging selective reading to mimic equivalent hash primitives.

\begin{table}[!t]
  \centering
  \caption{Performance overview of current hash functions  \protect \footnotemark}
  \label{tab:hash-comparision}
  \vspace{-0.3cm}
  \footnotesize
  \begin{tabular}{|l|c|c|c|c|}
    \hline
    \textbf{Hash functions} & \textbf{Key size}&\textbf{GE} & \textbf{Power} & \textbf{Clock cycles} \\
    \hline
     SHA-256\cite{feldhofer2006case} &256& 10,868 &15.87$\mu A$  & 1,128 \\
     \hline
     SHA-1 \cite{feldhofer2006case}& 160&8,120 & 10.68$\mu A$  & 1,274\\
     \hline
     AES \cite{feldhofer2004strong} &128&  3,400 &8.15$\mu A$  & 1,032 \\
     \hline
     MAME\cite{yoshida2007mame} & 256&8,100 &5.16$\mu A$  & 96 \\
     \hline
     MD5 \cite{feldhofer2006case}& 128& 8,400 & -  & 612 \\
     \hline
     MD4 \cite{feldhofer2006case} & 128& 7,350 & -  & 456 \\
     \hline
     PRESENT-80 \cite{bogdanov2007present} & 80& 1,570 & -  & 32 \\
     \hline
     PRESENT-80 \cite{rolfes2008ultra} & 80&1,075 & -  & 563 \\
     \hline
     PRESENT-128 \cite{bogdanov2008hash} & 128&1,886 & -  & 32 \\     
     \hline
     DES \cite{poschmann2007new} & 56 &2,309 & -  & 144 \\
     \hline
     mCrypton \cite{lim2005mcrypton} & 96&2,608 & -  & 13 \\
     \hline
     TEA \cite{yusecurity} & 128 &2,355 & -  & 64 \\
     \hline
     HIGHT \cite{hong2006hight} & 128 &3,048 & -  & 34 \\
     \hline
     DESXL \cite{poschmann2007new} & 184 &2,168 & -  & 144 \\
      \hline
     Grain \& Trivium\cite{good2007hardware} & 80 & 2,599 & -  & 1 \\
    \hline
  \end{tabular}
\end{table}

\footnotetext{ `-' means the algorithm is presented in theory and does not have  specific power consumption.}

\textbf{Design of hash enabled protocol.} To drive our key point, we conduct a brief survey of previous related works. We list several key usage scenarios that we would like to support. Our objective is not to complete the list, but to motivate our design.
(1) \emph{Cardinality estimation.} Dozens of estimators \cite{kodialam2007anonymous,sheng2008finding,sze2009fast,qian2010asap,qian2011cardinality,shah2011cardinality,shahzad2012every,zheng2013zoe,chen2013understanding,xiao2013differential,gong2014arbitrarily,liu2014fast,zheng2014towards,liu2015rfid,hou2015place} have been proposed in the past decade. For example, Qian \etal \cite{qian2011cardinality} proposed an estimation scheme called lottery frame. Shahzad and Liu \cite{shahzad2012every} estimated the number based on the average run-length of ones in a bit string received using the FSA. In particular, they claimed that their protocol is compatible with Gen2 systems. However, their scheme still requires modifying the communication protocol, and thus, it fails to work with COTS Gen2 systems. By contrast, our prototype can operate in COTS Gen2 systems as demonstrated in this study. (2) \emph{Missing detection.} The missing detection problem  was firstly mentioned in \cite{tan2008monitor}. Thereafter,  many follow-up works \cite{li2010identifying,luo2011efficient,zhang2011fast,luo2012probabilistic,zheng2015p,tan2010efficient,li2013efficient,liu2014multiple,luo2014missing,ma2012efficient,xie2014rfid,shahzad2015expecting,shahzad2016fast,yu2016missing,yu2015finding} have started to study the issue of false positives resulting from the collided slots  by using multiple bitmaps. Additional  details regarding  this application are introduced in $\S$\ref{section:usage}.
(3) \emph{Continuous reading.} The traditional inventory approach starts from the beginning each time it interrogates all the tags, thereby making it highly time-inefficient. These works \cite{sheng2010efficient,xie2013continuous,liu2014towards} have proposed continuous reading protocols that can incrementally collect tags in each step using the bitmap. For example, Sheng \etal \cite{sheng2010efficient} aimed to preserve the tags collected in the previous round and collect only unknown tags. Xie \etal \cite{xie2013continuous} conducted an experimental study on mobile reader scanning. Liu \etal \cite{liu2014towards} initially estimated the number of overlapping tags in two adjacent inventories and then performed an effective incremental inventory.
(4) \emph{Data mining.} These works \cite{sheng2008finding,xie2014efficiently,luo2013efficient,luo2016efficient,liu2015fast} discuss how to discover potential information online through bitmaps. For example,  Sheng \etal  \cite{sheng2008finding} proposed to identify the popular RFID categories using the group testing technique.  Xie \etal found histograms over tags through a small number of bitmaps\cite{xie2014efficiently}. Luo \etal \cite{luo2013efficient,luo2016efficient} determined whether the number of objects in each group was above or below a threshold. Liu \etal \cite{liu2015fast} proposed a new online classification protocol for a large number of groups.
(5) \emph{Tag searching.} These works\cite{liu2015step,zheng2013fast} have studied the tag searching problem that aims to find wanted tags from a large number of tags using bitmaps in a  multiple-reader environment. Zheng \etal  \cite{zheng2013fast} utilized bitmaps to aggregate a large volume of RFID tag information and to search the tags quickly.  Liu \etal \cite{liu2015step}  first used the testing slot technique to obtain the local search result by iteratively eliminating wanted tags that were absent from the interrogation region.
(6) \emph{Tag polling.} \cite{qiao2011energy,qiao2013tag,li2016locp} consider how to quickly obtain the sensing information from sensor-augmented tags. The system requires to assign a time slot to each tag using the presence bitmap. In summary, all the aforementioned HEP designs have allowed RFID research to develop considerably in the past decade.  All the work can be boosted by our hash primitives.

\section{Overview}
\label{section:overview}
%\begin{figure}[t!]
%  \centering
%  \includegraphics[width=0.9\linewidth]{architecture}
%   \caption{Gen2 communication protocol. \textnormal{Gen2 RFID systems adopts ALOHA-based air protocol. The reader divides the time into different slots. Each tag randomly selects a time slot to reply its ID. As a result, each time slot has three cases: empty, single or collision. }}
%  \label{fig:timing}
%\end{figure}

Tash is a software framework that provides a group of fundamental hash primitives for HEPs.  This section presents its usage scope and formally defines our problem domain.

\subsection{Scope}

Despite clear and certain specifications, the implementation of the Gen2 protocol still varies with readers and manufacturers because of firmware  bugs or compromises,  especially in early released reader devices, according to our compatibility report presented in  $\S$\ref{section:microbenchmark}.  Here, we firmly claim
that our design and implementation strictly follow the specifications of the Gen2 and LLRP protocols  (refer to $\S$\ref{section:implementation}). The framework  works with any Gen2-compatible readers and tags. The  performance losses caused by  defects in devices are outside the scope of our discussion.

\subsection{Definitions of Hash Primitives}

%\begin{figure}[t!]
%  \centering
%  \includegraphics[width=0.9\linewidth]{tash-table}
%   \caption{Illustration of tashing and tashing table functions. \textnormal{There are total 7 tags in the scene. Their IDs are $d_1,d_2,\dots$ and $d_7$. After tashing, these tags are tashed into a 3-bit tashing table $B=[0, 0, 3, 0, 1, 1, 2, 0]$ with a random number $r$. For instance, $f_3(d_1,r)=2$ and $f_3(d_4,r)=4$.}}
%  \label{fig:tasing}
%\end{figure}

Before delving into  details, we formally define the hash primitives that the HEPs require, from a high-level.

\begin{defn}[Tash function]
\label{defn:tash-function}
An $l$-bit \emph{tash function} is actually a hash function $f_l(t,r): \mathcal{T}\times \mathcal{R}\rightarrow 2^l$, where $\mathcal{T}$ and $\mathcal{R}$ are the domains of  \texttt{EPC}s of the tags and random seeds. 	
\end{defn}

\textbf{Tash function and tash value.} As the above definition specifies, an $l$-bit tash function takes an \texttt{EPC} $t$ and a random seed $r$ as input and outputs an $l$-bit integer $i$, denoted  by:
\begin{equation}
	i=f_l(t,r)
\end{equation}
We call $l$ the \emph{dimension} of tash function (\ie $l=0,1,2,\dots$). The tash value $i$ is an integer $\in [0,2^l-1]$. Similar to other common hash function, the tash function has three basic characteristics. First, the output changes significantly when the two parameters are altered.  Second, its output is uniformly distributed within the given range, and predicable if all inputs are known. Third, the hash values are accessible.

%Like other common hashing function, tash function has three basic characteristics.  First, it takes a tag's \texttt{EPC} and a random seed  as input. The output would change a lot when the two parameters are altered.  Second, its output is randomly distributed within the given range, and predicable if knowing all inputs. Third, the computation is one-way and irreversible, \ie output reveals nothing about input.

\begin{defn}[Tash table function]
\label{defn:tash-table}
An $l$-bit tash table function can assign each tag $t$ from a given set into the $i^{th}$ entry of a hash table (aka tash table) with a random seed $r$, where $i=f_l(t,r)$. Each entry of the tash table is the number of tags tashed into it.
\end{defn}

\textbf{Tash table function and tash table.} Let $B$ and $\mathcal{F}_l$ denote a  tash table and a tash table function respectively. The tash table function takes a set of tags (\ie $T=\{t_1,t_2,$ $\dots, t_n\}$) and a random number $r$ as  input and  outputs a tash table $B$, denote by:
\begin{equation}
	B=\mathcal{F}_l(T,r)
\end{equation}
where $B[i]=|\{t|f_l(t,r)=i\}|$ (\ie the number of tags tashed into the $i^{th}$ entry) for $\forall t\in T$.  Let $L=2^l$, which is defined as the \emph{size} of the tash table. The tash table function is the core function that HEPs expect.  HEPs consider the reader as well as all tags as a black box equipping with tash table function. When inputing a random seed, the box would output a tash table. HEPs then utilize such table to provide various services (\eg missing detection or cardinality estimation.). It worths noting that superior to the bitmap employed in prior HEPs, our tash table is a perfect table that contains the exact number of tags tashed into each entry.  Clearly, the table is completely backward compatible with prior HEPs because it can be forcedly converted into a presence bitmap.

\begin{figure}[t!]
  \centering
  \includegraphics[width=0.8\linewidth]{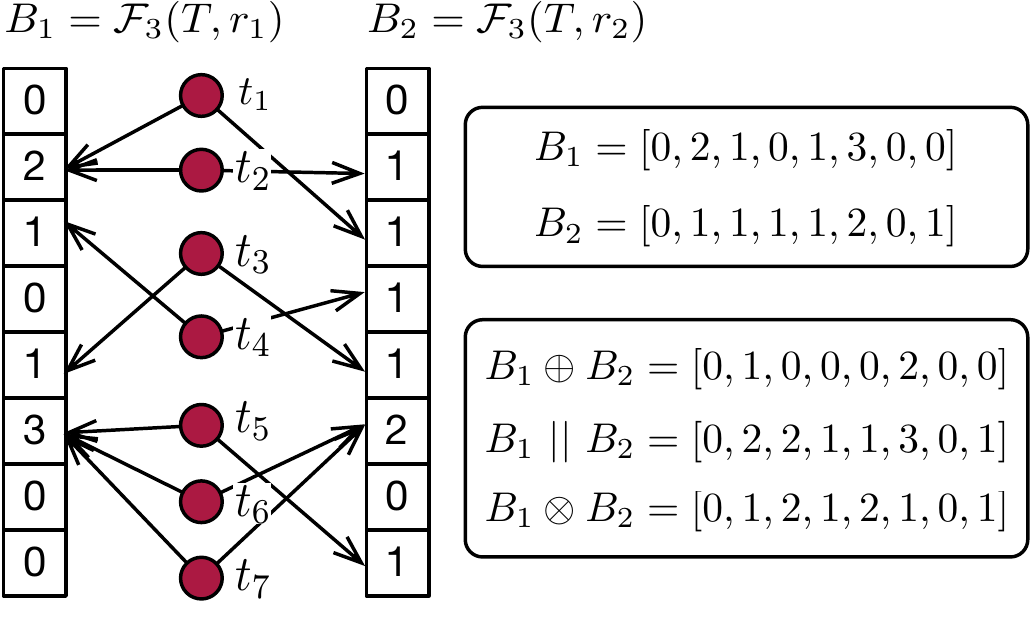}
   \caption{Illustration of tash operators. \textnormal{The left  shows two independent tash tables,  while the right  shows the results of the two tash tables with tash AND, OR and XOR.}}
  \label{fig:tash-operator}
    \vspace{-0.3cm}
\end{figure}

\textbf{Tash operators.} Most prior HEPs adopt probabilistic ways and their results are guaranteed with a given confidence level. To meet the level, they usually combine multiple bitmaps, which are acquired through multiple rounds and challenged by different seeds. We abstract such combination into three basic tash operators, namely, tash AND, OR and XOR. These operators can comprise other complex operations. Let $B_1= \mathcal{F}_l(T, r_1)$ and $B_2=\mathcal{F}_l(T, r_2)$ denote two tash tables acquired twice with two different seeds, $r_1$ and $r_2$.

\begin{defn}[Tash AND]
\label{defn:tash-and}
The tash AND  (denoted by $\oplus$) of two tash tables is to obtain the intersection of  two corresponding entry sets. Formally,  $B =B_1\oplus B_2$, where $B[i]=|\{t|f_l (t,r_1)=i\& f_l(t,r_2)=i\}|$.
\end{defn}
The tash AND is aimed at obtaining the common intersection of corresponding entries from two tash tables. For example,  as shown in Fig.~\ref{fig:tash-operator},  $B_1[1]$ and $B_2[1]$ count  $\{t_1,t_2\}$ and $\{t_2\}$ respectively. However, $(B_1\oplus B_2)[1]=|\{t_2\}|=1$, which counts $t_2$ only.

\begin{defn}[Tash OR]
\label{defn:tash-or}
The tash OR (denoted by $||$) of two tash tables is to merge two corresponding entry sets. Formally, $B = B_1|| B_2$, where $B[i]=|\{t | f_l(t,r_1)$ $=i||f_l(t,r_2)=i \}|$.
\end{defn}
The tash OR is aimed at obtaining the total number of tags mapped into the corresponding entries in two tash tables.  Note tash OR is not the same as the entry-wise sum, \ie $B_1 || B_2 \neq B_1+B_2$ because the tags twice mapped into a same entry are counted only once.  As shown in Fig.~\ref{fig:tash-operator}, $(B_1||B_2)[5]=|\{t_5,t_6,t_7\}|=3$ although $B_1[5]+B_2[5]=5$ because  $t_6$ and $t_7$ appear twice in the two tash tables.

\begin{defn}[Tash XOR]
\label{defn:tash-xor}
The tash XOR (denoted by $\otimes$) is to remove the intersection of two corresponding entry sets from the first entry set. Formally, $B = B_1 \otimes B_2$ such that $B[i]=|\{t| f_l(t,r_1)=i\ \&\ f_l(t,r_2)\neq i\}|$.
\end{defn}
The tash XOR is aimed at obtaining the total number of the set difference. As Fig.~\ref{fig:tash-operator} shows, $B_1[5]=|\{t_5,t_6,t_7\}|$ and $B_2[5]=|\{t_6, t_7\}|$. Then $(B_1 \otimes B_2)[5]=|\{t_5\}|=1$.

%\begin{defn}[Tash NOT]
%The tash NOT is defined as $B=\overline{\mathcal{F}}_L(T, r)$ such that $B[i]=|\{t| f_l(t,r)\neq i\}|$, \ie the number of tags which are not mapped into the $i^{th}$ entry.
%\end{defn}
%The tash NOT is a monadic operator, targeting to obtain the remaining tags which are not mapped into the corresponding entry.

The above operators can be applied in a series of tash tables with the same dimension for a hybrid operation, \eg $B_1\oplus B_2 || B_3$. Tash AND and OR  satisfy operational laws such as associative law and commutative law, \eg $B_1\oplus B_2 = B_2 \oplus B_1$.  The design of tash operators is one of the attractive features of the tash framework, and it has never been proposed before. More importantly, we design and implement these operators in the physical layer to provide one-stop acquisition solution.

\subsection{Solution Sketch}

Tash is designed to reduce the overhead for air communications. It runs in the middle of the reader and upper application. The upper application passes a pair of arguments (\ie $r$ and $l$), or pairs of arguments (as well as operators) to \oursystem.  On the basis of the arguments, \oursystem generates one or more configuration files to manipulate the reader's reading. Finally, \oursystem abstracts the reading results to a tash table, which is returned to the upper application.

The rest of the paper is structured as follows.  We firstly present the tash design in $\S$\ref{section:design}. We next demonstrate the usage of our hash primitives in two classic applications in $\S$\ref{section:usage}. We then introduce the tash implementation using LLRP interfaces in $\S$\ref{section:implementation}. In $\S$\ref{section:microbenchmark} and $\S$\ref{section:usage-evaluation}, we present the microbenchmark and the usage evaluation. Finally, we conclude in $\S$\ref{section:conclusion} and present future directions.

%The next three sections elaborate the above issues in depth and provide technical details.

\section{Tash Design}
\label{section:design}

In this section, we introduce the background of selective reading in Gen2, and then present the technical details of our designs.

\begin{figure}[t!]
  \centering
  \includegraphics[width=0.7\linewidth]{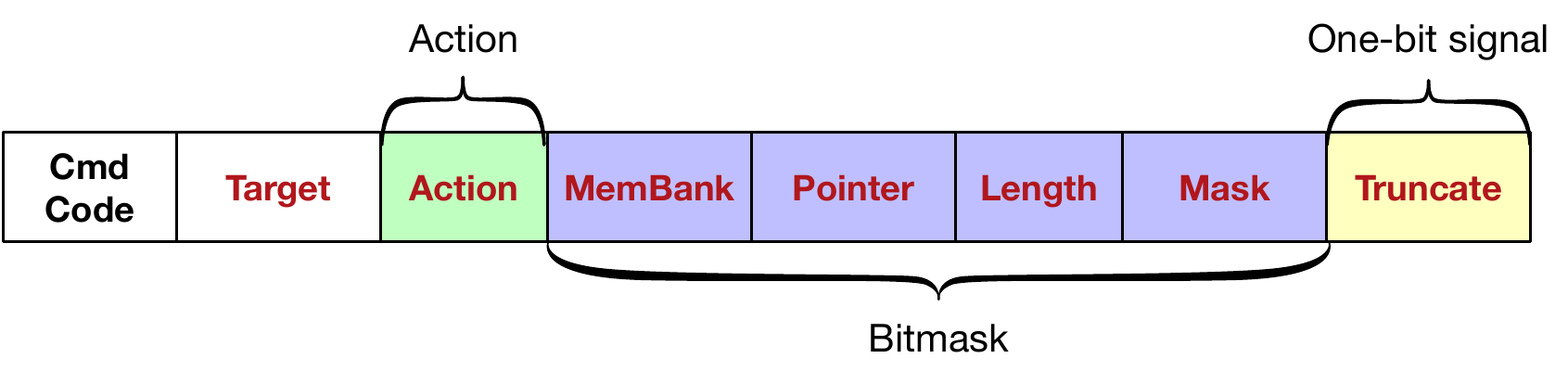}
   \vspace{-0.6cm}
   \caption{Formats of \texttt{Select} commands. }
  \label{fig:gen2-select}
    \vspace{-0.3cm}
\end{figure}

\subsection{Background of Gen2 Protocol}

The Gen2 standard defines air communication between readers and tags.  On the basis of \cite{gen2,zhang2009rfid}, we introduce its four central functions we will employ:

\textbf{F1: Memory Model.}
Gen2 specifies a simple tag memory model (pages 44 $\sim$ 46 of \cite{gen2}). Each tag contains four types of non-volatile memory blocks (called \emph{memory banks}): (1) \texttt{MemBank-0} is reserved for password associated with the tag. (2) \texttt{MemBank-1} stores the \texttt{EPC} number. (3) \texttt{MemBank-2} stores the \texttt{TID} that specifies the unchangeable tag and vendor specific information. (4) \texttt{MemBank-3} is a customized storage that contains user-defined data.

\textbf{F2: Selective Reading.}
Gen2 specifies that each inventory must be started with \texttt{Select} command (pages 72$\sim$73 of \cite{gen2}). The reader can use this command to choose a subset of tags that will participate in the upcoming inventory round. In particular, each tag maintains a flag variable \texttt{SL}. The reader can use the \texttt{Select} command to turn the \texttt{SL} flags of tags into \texttt{asserted} (\ie true) or \texttt{deasserted} (\ie false).  The \texttt{Select} command comprises six mandatory fields and one optional field apart from the constant cmd code (\ie $1010_2$), as shown in Fig.~\ref{fig:gen2-select}. The following fields are presented for this study.

%====Reserved for trans

\begin{table}[b!]
  \centering
  \small
  \caption{Actions in the \texttt{Select} command}
  \vspace{-0.3cm}
   \label{tab:action}
  \begin{tabular}{|c|c|c|}
    \hline
    Action code & Tag matching & Tag not-matching\\
    \hline
    $0$ & assert \texttt{SL} & deassert \texttt{SL}\\\hline
    $1$ & assert \texttt{SL} & do nothing\\\hline
    $2$ & do nothing & deassert\texttt{SL}\\\hline
    $3$ & negate \texttt{SL} & do nothing\\\hline
    $4$ & deassert \texttt{SL} & assert \texttt{SL}\\\hline
    $5$ & deassert \texttt{SL} & do nothing\\\hline
    $6$ & do nothing & assert \texttt{SL}\\\hline
    $7$ & do nothing & negate \texttt{SL}\\\hline
  \end{tabular}
\end{table}

$\bullet$ \texttt{Target}. This field allows a reader to change  \texttt{SL} flags or the inventoried flags of the tags. The inventoried flags are used when multiple readers are present. Such scenario is irrelevant to our requirements. Thus, we aim to change \texttt{SL} flags only by setting \texttt{Target=$100_2$}.

$\bullet$ \texttt{Action}. This field specifies an action that will be will performed by the tags. Table.~\ref{tab:action} lists eight action codes to which the tag makes different responses. For example,  the matching or not-matching tags assert or deassert their \texttt{SL} flags  when \texttt{Action=$0$}. We  leverage this useful feature to design tash operators.

$\bullet$ \texttt{MemBank}, \texttt{Pointer}, \texttt{Length} and \texttt{Mask}.  These four fields are combined to compose a \emph{bitmask}. The bitmask indicates which tags are matched or not-matched for an \texttt{Action}. The \texttt{Mask} contains a variable length binary string that should match the content of a specific position in the memory of a tag. The \texttt{Length} field defines the length of the \texttt{Mask} field in bits. The \texttt{Mask} field can be compared with one of the four types of memory banks in a tag. The \texttt{MemBank} field specifies which memory bank the \texttt{Mask} will be compared with. The \texttt{Pointer} field specifies the starting position in the memory bank where the \texttt{Mask} will be compared with. For example, if we use a tuple $(b,p,l, m)$ to denote the four fields, then only the tags with data starting at the $p^{th}$ bit with a length of $l$ bits in the $b^{th}$ memory bank that is equal to $m$ are matched.

\begin{figure}[t!]
  \centering
  \includegraphics[width=0.8\linewidth]{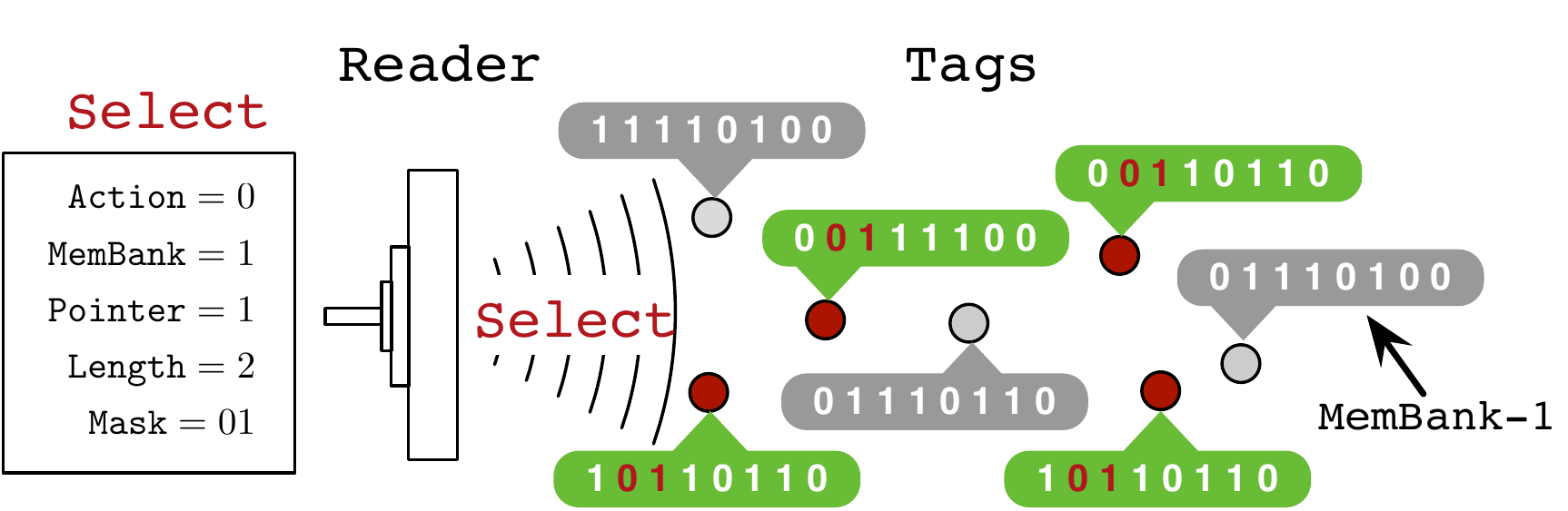}
   \caption{Illustration of selective reading in Gen2. \textnormal{There are total $7$ tags covered by a reader. The reader initiates a selective reading using a \texttt{Select} command with parameters: \texttt{Action}=0, \texttt{MemBank}=1, \texttt{Pointer=1}, \texttt{Length=1}, \texttt{Mask}=01. This command means that these tags (highlighted with dark red) whose data starting at the first bit with a length of $2$ bits in the \texttt{MemBank-1} equals $01_2$ are selected to participate in the incoming inventory, while other tags (with gray color) that do not meet the condition remain silent. As a result, only $4$ tags are collected in this round of inventory.}}
  \label{fig:selective-reading}
    \vspace{-0.2cm}
\end{figure}

To visually understand the selective reading, we show an example in Fig.~\ref{fig:selective-reading} in which 4 out of 7 tags are selected to participate in the incoming inventory. Complex and multiple subsets of tags can be facilitated by issuing a group of \texttt{Select} commands to choose a subset of tags before an inventory round starts. For example, we can issue two \texttt{Select} commands: one  for division and another for one-bit reply. Note the \texttt{Truncate} enabled \texttt{Select} command  must be the last one if multiple selection commands are issued \cite{gen2}.

\textbf{F3: Truncated Reply.} Gen2 allows tags to reply a \emph{truncated} reply (\ie replying a part of  \texttt{EPC}) through a special \texttt{Select} command with an enabled \texttt{Truncate} field, making a one-bit presence signal possible. When \texttt{Truncate} is enabled (\ie set to $1$), then the corresponding bitmask is not used for the division of tags, but lets tags reply with a portion of their \texttt{EPC}s following the pattern specified by the bitmask. Note that when \texttt{Truncate} is enabled, the \texttt{MemBank} must be set to the \texttt{EPC} bank (\ie $\texttt{MemBank}=1$) and such \texttt{Select} command must be the last one.

\textbf{F4: Query Model.}
Followed by a group of \texttt{Select} commands,  \texttt{Query} command (see page 76$\sim$80 of \cite{gen2}.) starts a new \emph{inventory round} over a subset of tags, chosen by the previous \texttt{Select} commands. There are $7$ fields in the \texttt{Query} command. We only focus on \texttt{Sel} field, which is most tightly relevant to the selective reading. As mentioned above, the \texttt{Select} command has divided the tags into two opposite subsets with asserted and deasserted \texttt{SL} respectively. The \texttt{Sel} field specifies which subset will reply in the current inventory round. If \texttt{Sel=$11_2$}, the tags with asserted \texttt{SL} reply. If \texttt{Sel=$10_2$}, the tags with deasserted \texttt{SL} reply. We choose the tags with asserted \texttt{SL} by default.

\subsection{Design of the Tash Function}

An $l$-bit tash function is essentially a hash function that is indispensable to HEPs. We design the tash function while following the three principles outlined as follows. The first principle requires that the tash result must be dependent on the input \texttt{EPC} and the seed. Moreover, it must be predictable as long as all the input parameters are known. The second principle requires the output values to be random, \ie uniformly distributed in $[0, 2^l-1]$. Even a one bit difference in the input will result in a completely different outcome. The third principle requires a method that can access the tash result of a tag directly or indirectly.

We have constructed the tash function as follows by applying the aforementioned principles:  given a tag with an \texttt{EPC} of $t$, we firstly calculate the \emph{hash value} of the \texttt{EPC} offline, using a common perfect hash function like $128$-bit \texttt{MD5} or \texttt{SHA-1}. Let $h(t)$ denote the calculated hash value. We then write  $h(t)$ into the tag's user-defined memory bank of the tag, \ie \texttt{MemBank-3}, for later use.
\begin{defn}[Tash value]
\label{defn:tash-value}
	The $l$-bit tash value of  tag $t$ challenged by  seed $r$ is defined as the value of the sub-bitstring starting from the $r^{th}$ bit and ending at the $(r+l-1)^{th}$ bit in the \texttt{MemBank-3} of the tag.
\end{defn}

Evidently, $f_l(t,r)$ is actually a portion of $h(t)$, and thus,  the parameter $r\in [0,\mathcal{L}-1]$ and $l\in [1, \mathcal{L}-r]$, where $\mathcal{L}$ is the length of the hash value (\eg $128$ bits).  Fig.~\ref{fig:tash-function-example} shows an example wherein the \texttt{MemBank-1} and \texttt{MemBank-3} of the tag store its \texttt{EPC} $t$ and the hash value $h(t)$, respectively. When $r=5$ and $l=4$ are inputted,  the tash value that this tag outputs is $1010_2$, which is the sub-bitstring of $h(t)$ starting from the $5^{th}$ bit and ending at the $8^{th}$ bit in \texttt{MemBank-3}, \ie $f_4(t,5)=1010_2$. 

Our design does not require a tag to equip a real hash function or the engagement of its chip. It clearly applies the preceding principles. First,  $f_l(t,r)$ is evidently repeatable, predicable and dependent on the inputs. Second, the randomness of $f_l(t,r)$ is derived from $h(t)$ and $r$, which are supposed to have a good randomness quality. Third,  we have two ways to access the tash value. We can use the memory \texttt{Read} command to access  \texttt{MemBank-3} of a tag directly, or use the selective reading function to access the tash value indirectly (discussed later).

\begin{figure}[t!]
  \centering
  \includegraphics[width=0.8\linewidth]{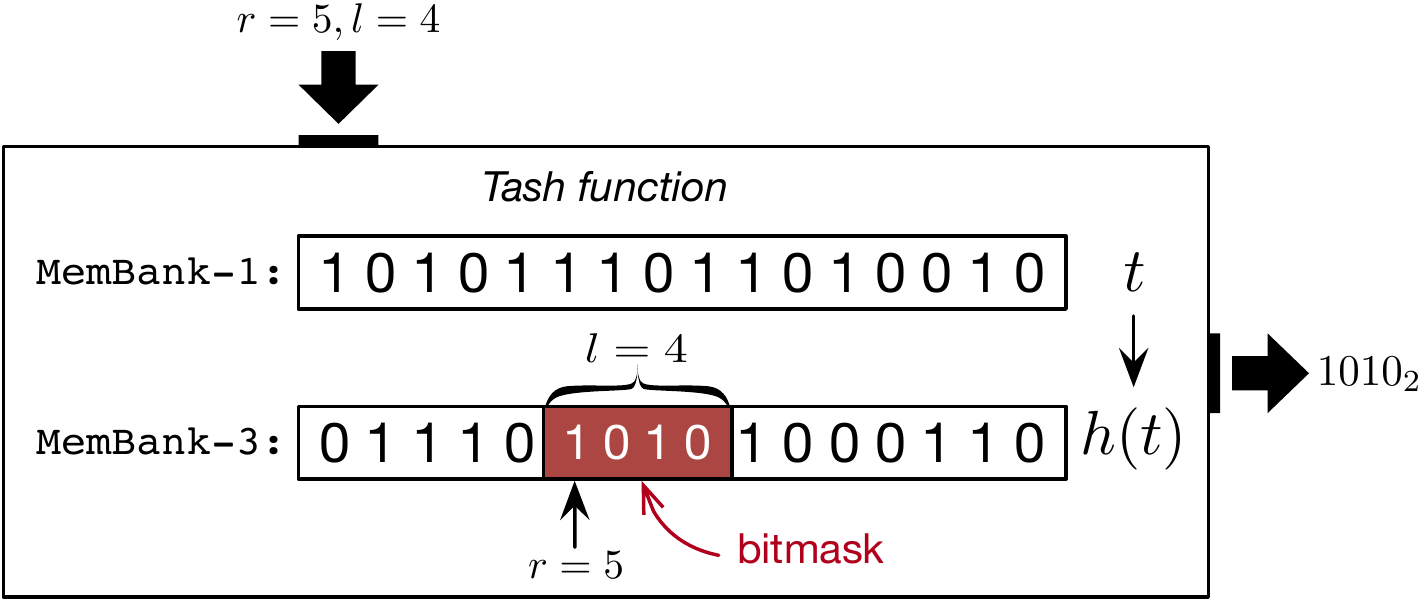}
   \caption{Illustration of a tash function. \textnormal{The result of a tash function of $f_4(t, 5)$ is equal to $1010_2$.}}
  \label{fig:tash-function-example}
    \vspace{-0.3cm}
\end{figure}

\textbf{Discussion:}  A few points are worth-noting about the design:

$\bullet$ As the tash value is a portion of the hash value, if two random numbers may cover a common sub-string. For example, if $r_1$ and $r_2$ differ by $1$, there exist $l-1$ same bits with $50\%$ of probability that two hash values are same,  although such case occurs with a small probability, \ie $\approx 127/(128\times 128)\times 0.5=0.0039$. If some upper applications require extremely strong independence, we should generate the second random number $r_2$ meeting the condition of $r_2<r_1-l$ and $r_2\geq r+l$, so as to avoid the common coverage and potential relevance.

$\bullet$ The design of tash function involves the \texttt{MemBank-3}, \ie the user-defined storage. We can use \texttt{Write} command to store any data into this memory bank. Our compatibility report (shown in $\S$\ref{section:microbenchmark}) suggests that almost all types of tags support both \texttt{MemBank-3} and \texttt{Write} command except one read-only type (\ie ImpinJ Monza R6). Our approach is generally practicable.

$\bullet$ Our design targets at enabling COTS tags, billions of which have been deployed in recent years,  to obtain performance advantages from well-studied hash based protocols, instead of enhancing their security or privacy preservation. Our design still follows the current COTS tag's security mechanism, \ie password protected memory access.

$\bullet$ Tash function also offers a good feature that the computation is one way and irreversible, \ie the output reveals nothing about the input. This feature is inherited from the hash function. It may be useful for privacy protection in practice. However, this topic is beyond the scope of this work.

\subsection{Design of the Tash Table Function}
The tash table function treats a reader and multiple tags as if they were a single virtual node, outputting a tash table.  For simplicity, we use
$$
S ( \underbrace{a}_{\texttt{Action}}, \overbrace{b}^{\texttt{MemBank}}, \underbrace{p}_{\texttt{Pointer}}, \overbrace{l}^{\texttt{Length}}, \underbrace{m}_{\texttt{Mask}}, \overbrace{u}^{\texttt{Truncate}})
$$
to denote a selection command (\ie \texttt{Select}) with an \texttt{Action} $(a)$, a \texttt{MemBank} $(b)$, a \texttt{Pointer} $(p)$, a \texttt{Length} $(l)$, a \texttt{Mask} $(m)$ and a \texttt{Truncate} $(u)$. The command aims to select a subset of  tags with a sub-bitstring that starts from the $p^{th}$ bit and ends at the $(p+l-1)^{th}$ bit in the $b^{th}$ memory bank that is equal to $m$. These selected tags are requested to take an action $a$. The action codes are shown in Table.~\ref{tab:action}. In particular, if $u=1$, then each tag will reply with a truncated \texttt{EPC} number.

The tash table function  is designed as follows.  An $l$-bit table $B$ consists of a total of  $2^l$ entries, each of which contains the amount of tags mapped into it. In particular, the index number of each entry, which ranges from $0$ to $2^l-1$, is actually the tash values of the tags mapped into this entry, \ie $B[\textcolor{dark-red}{\underline{i}}]=|\{t|f_l(t,r)=\textcolor{dark-red}{\underline{i}}\}|$. When constructing the $i^{th}$ entry, the reader performs selective reading with two selection commands as follows:
$$
S_1(0, 3, r, l, i, 0) \text{ and } S_{\Box}(1, 1, 1, 1, 1, 1)
$$
Command $S_1$ selects a subset of  tags with a sub-bitstring that starts from the $r^{th}$ bit and ends at the $(r+l-1)^{th}$ bit in the \texttt{MemBank-3} that is equal to $i$.  Notably, the involved sub-bitstring is the tash value of a tag, \ie $f_l(t,r)$, which refers to Definition.~\ref{defn:tash-value}. Consequently, only  tags with tash values equal to $i$ are selected to participate in the incoming inventory, \ie counted by the $i^{th}$ entry.  The second command $S_\Box$ enables the selected tags to reply with the first bit of their \texttt{EPC} numbers for the one-bit signals. We call such inventory round as an \emph{entry-inventory}. In this manner, we can obtain the whole tash table by launching $2^l$ entry-inventories.

To visually understand the procedure, we illustrate an example in Fig.~\ref{fig:tash-table-function-example}, where $r=5$ and $l=2$. The tash table contains $2^2$ entries; hence,  four entry-inventories are launched. Their selection commands are defined as follows:
\begin{equation*}
\begin{split}
\text{\ding{182} } S_1(0, 3, \underline{5}, \underline{2}, \textcolor{dark-red}{\underline{0}}, 0) \text{ and } S_{\Box}(1, 1, 1, 1, 1, 1)\\
\text{\ding{183} }S_1(0, 3, \underline{5}, \underline{2}, \textcolor{dark-red}{\underline{1}}, 0) \text{ and } S_{\Box}(1, 1, 1, 1, 1, 1)\\
\text{\ding{184} }S_1(0, 3, \underline{5}, \underline{2}, \textcolor{dark-red}{\underline{2}}, 0) \text{ and } S_{\Box}(1, 1, 1, 1, 1, 1)\\
\text{\ding{185} }S_1(0, 3, \underline{5}, \underline{2}, \textcolor{dark-red}{\underline{3}}, 0) \text{ and } S_{\Box}(1, 1, 1, 1, 1, 1)
\end{split}
\end{equation*}
For the third entry-inventory, the \texttt{Mask} field is set to $2$ because the index of the third entry is $2$.  Four tags (\ie $t_5$, $t_6$, $t_7$ and $t_8$) are selected to join in this entry-inventory. Thus,  $\mathcal{F}_2(T,5)[2]=4$.

For a tash table, note that ($1$) the sum of all its entries is equal to the total number of tags, and ($2$) it allows an application to selectively construct the entries of a tash table becaues each entry-inventory are independent of each other and completely controllable. For example, we can skip the inventories of these entries that are predicted to be empty.

\subsection{Design of  Tash Operators}

\begin{figure}[t!]
  \centering
  \includegraphics[width=0.8\linewidth]{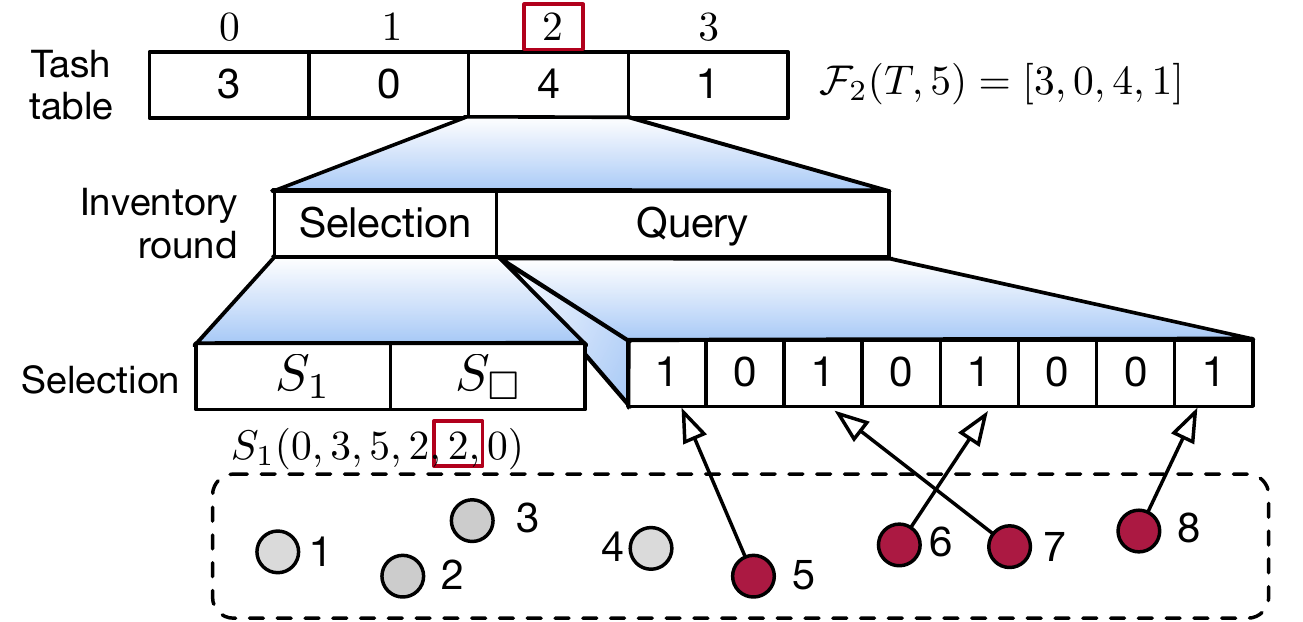}
   \caption{Illustration of creating a tash table.
   \textnormal{Given that $r = 5$ and $l = 2$, $\mathcal{F}_2(T,5) = [3,0,4,1]$. Zooming into the $3^{rd}$ entry-inventory,  tags $t_5,t_6,t_7$ and $t_8$ are selected to join the inventory. $S_{\Box}$ means this is the end command.}   }
  \label{fig:tash-table-function-example}
    \vspace{-0.3cm}
\end{figure}

A tash operator is connected to two tash tables, which have the same dimensions but are constructed using two different seeds. When two seeds, $r_1$ and $r_2$, are given,  we can obtain two $l$-bit tash tables: $B_1=\mathcal{F}_l(T,r_1)$ and $B_2=\mathcal{F}_l(T,r_2)$. Our objective is to obtain a final tash table $B$ by performing one of the subsequent tash operators on $B_1$ and $B_2$.

\textbf{Tash AND}. If $B=B_1\oplus B_2$, then each entry of $B$ denotes  the number of tags that are concurrently mapped into the corresponding entries of  $B_1$ and $B_2$. The selection commands for the $i^{th}$ entry-inventory are defined as follows:
$$
S_1(\textcolor{dark-red}{\underline{0}},3,r_1,l,i,0), S_2(\textcolor{dark-red}{\underline{2}},3,r_2,l,i,0), S_{\Box}
$$
From the action codes shown in Table.~\ref{tab:action}, the purpose of $S_1$ with action code of $0$  is to select tags $\in B_1[i]$ and deselect tags $\notin B_1[i]$. $S_2$ with action code of $2$  deselects tags  $\notin B_2[i]$ and results in tags $\in B_2[i]$  doing nothing.  After $S_1$ is received, each tag exhibits one of two states, \ie selected or deselected. Then, $S_2$ will make the selected tags remain in their selected states if they match its condition (\ie doing nothing); otherwise, it changes their states to the deselected states (\ie selected $\rightarrow$ deselected), which is equivalent to removing tags $\notin B_2[i]$ from tags $\in B_1[i]$. Meanwhile, the tags deselected by $S_1$ remain in their states regardless of whether they match (\ie do nothing) or not match (\ie deselected $\rightarrow$ deselected) the condition of $S_2$. Finally, $S_{\Box}$ is reserved for the one-bit presence signal.

\textbf{Tash OR}. If $B = B_1 || B_2$, then each entry of $B$ is the number of tags that mapped into the corresponding entry of either $B_1$ or $B_2$. The selection commands for the $i^{th}$ entry-inventory are defined as follows:
$$
	S_1(\textcolor{dark-red}{\underline{0}},3,r_1,l,i,0), S_1(\textcolor{dark-red}{\underline{1}},3,r_2,l,i,0),S_{\Box}
$$
Similarly,  $S_1$ selects tags $\in B_1[i]$ and deselect tags $\notin B_1[i]$.  $S_2$ with action code of $1$ (see Table.~\ref{tab:action}) allows tags $\in B_2[i]$ to be selected as well, but tags $\notin B_2[i]$ remain in their states (\ie do nothing), some of  these tags may have been selected by $S_1$. The process is equivalent to holding the tags selected by $S_1$ and incrementally  adding the new tags selected by $S_2$.

\textbf{Tash XOR}. If $B = B_1 \otimes B_2$, then each entry of $B$ is the number of tags that are mapped into the corresponding entry of $B_1$ but not into the entry of $B_2$. The selection commands for the $i^{th}$ entry-inventory are defined as follows:
$$
	S_1(\textcolor{dark-red}{\underline{0}},3,r_1,l,i,0), S_2(\textcolor{dark-red}{\underline{5}},3,r_2,l,i,0),S_{\Box}
$$
Similarly,  $S_2$ allows tags $\in B_2[i]$ to be deselected (\ie removed from tags $\in B_1[i]$) and tags $\notin B_2[i]$ to do nothing. This process is equivalent to removing tags $\in B_2[i]$ from tags $\in B_1[i]$.

%\textbf{Tash NOT}. If $B=\overline{\mathcal{F}}_l(T, r)$, then $B[i]=|\{t| f_l(t,r)\neq i\}|$. Namely, each entry in the final table is the number of tags not-mapped into this entry in the $r$ table. Only one extra selection command is required for the $i^{th}$ round inventory:
%$$
%	S_1(100_2,3,r_1,l,m_i,0), S_{\Box}
%$$
%It should be noted that the operator of Tash NOT is an exception, which must be placed in the first position. Because the selection commands must execute in sequence, we could not store any intermediate results for the NOT when it appears in the none first positions.

\textbf{Tash hybrid.} The aforementioned  three operators can be further applied to a hybrid operation. When $k$ seeds (\ie $r_1,\cdots, r_k$) are given, we can obtain $k$ tash tables. The selection commands for the $i^{th}$ entry-inventory can be designed as follows:
\begin{eqnarray*}
	&& S_1(0,3,r_1,l, i,0), S_2(\text{\texttt{\texttt{AC}}},3,r_2,l, i,0),\\
	&& \cdots, S_k(\text{\texttt{\texttt{AC}}},3,r_k,l, i,0),S_{\Box}
\end{eqnarray*}
where \texttt{AC} represents the \texttt{Action} code, which is set to $2$, $1$ and $5$ for tash AND, OR and XOR, respectively. The action code of the first command is always set to $0$.  For example, the selection commands in the $i^{th}$ entry-inventory for $\mathcal{F}_l(T, r_1) \oplus \mathcal{F}_l(T, r_2)  || $ $\mathcal{F}_l(T,r_3) \otimes \mathcal{F}_l(T,r_4)$ are given by:
\begin{eqnarray*}
	&&S_1(\underline{0},3,r_1,l, i,0), S_2(\underline{2},3,r_2, l, i,0), \\
	&&S_3(\underline{1},3,r_3,l, i,0), S_4(\underline{5},3,r_4, l, i,0),S_{\Box}
\end{eqnarray*}
We leverage the action of a selection command to perform an operation in the physical layer before an entry-inventory starts, therefore, we introduce minimal additional communication overhead, \ie broadcasting multiple \texttt{Select} commands. Compared with the multiple acquisitions of bitmaps used by prior HEPs, our solution provides a one-stop solution that can significantly reduce the total overhead in such situation.

\subsection{Discussion}

\textbf{Comparison with bitmap.} A tash table evidently takes a considerably longer time to obtain than a bitmap because a bitmap requires only one round of inventory, whereas a tash table requires multiple rounds. The additional time consumption is the trade-off for practicality because the reply of a COTS tag at the slot level is out of control. Nevertheless, this additional cost brings an additional benefit, \ie a tash table has the exact number of tags mapped onto its each entry, which cannot be suggested by a bitmap. Moreover, a one-stop operator service can save more time.

\textbf{Embedded pseudo-random function.} Qian \etal \cite{qian2010asap} and Shahzad \etal \cite{shahzad2012every} proposed a similar concept of utilizing a pre-stored random bit-string to construct a lightweight pseudo-random function. These studies have inspired our work.  However, their main objective of these previous researchers is to accelerate the calculation of a random number, which still requires the engagement  with the chip of a tag, and thus,  has never been implemented in practice. In the present work, we do not require additional efforts on changing the logics of  a tag chip and we associate this concept with the function of selective reading, moving the main task from a tag to a reader. Our design not only preserves the good features of the hash function but also gives a practical solution. This process has never been performed before.

\textbf{Channel error.} Channel error is one of the most notorious problems of HEPs because pure one-bit signal transmission is vulnerable to ambient interference. Thus, an additional error control mechanism is expected to be applied to HEPs. In the Gen2 protocol, the CRC8 code is automatically appended to the data transmitted between a reader and a tag for error detection, even when one bit of EPC is transmitted. The corrupt data will be retransmitted. Therefore, we should not be concerned with channel error.

\section{Tash Usage}
\label{section:usage}

This section revisits two classic problems of HEPs for usage study. We propose two practical solutions that use tash primitives for these problems. Note that in spite of two demonstration presented in this section, our tash primitives especially the tash table can serve any kind of HEPs. 

\subsection{Usage I: Cardinality Estimation}

Cardinality estimation aims to estimate the total number of tags by using one-bit presence signals that are received without collecting each individual tag. The problem  is formally defined as follows:

\begin{problem}
\label{problem:cardinality}
When a tag population of an unknown size $n$,  a tolerance of $\beta\in (0,1)$, and a required confidence level of $\alpha\in(0,1)$ is given, how can the number of tags $\widehat{n}$ be estimated such that $\Pr(|\widehat{n}-n|\leq \beta n)\geq \alpha$?
\end{problem}

A naive method would be to add all the entries of a tash table together or let all tags reply at the first entry.  Since each tag participates in one and only one entry-inventory,  the final number is exactly equal to $n$.  Keep in mind that our each entry corresponds to a complete round inventory. The naive method is equivalent to collecting them all, which is extremely time-consuming. We subsequently provide a reliable solution in a probabilistic way.

\textbf{Proposed Estimator}. We leverage the number of tags mapped into the first entry of a tash table to estimate $n$. Let $X$ be the random variable to indicate the value of the first entry of a tash table. Since $n$ tags are randomly and uniformly assigned into $2^l$ entires, we have
\begin{equation}
	\Pr(X=m) = \binom{n}{m}p^m (1-p)^{n-m}
\end{equation}
where $p=1/2^l$. Evidently, variable $X$ follows a standard Binomial distribution with the parameters $n$ and $p$, \ie $X\sim B(n,p)$. Therefore the expected value $\mu=np$ and variance $\delta=np(1-p)$. By equating the expected value and an instanced value $m$, our estimator $\widehat{n}$ is given by:
\begin{equation}
	\widehat{n}= m/p = m 2^l
\end{equation}
The estimator only requires the first entry of the hash table, so it skips inventories of other entries. We must choose an appropriate $l$ to ensure the estimation error within the given tolerance level $\beta$ with a confidence of greater than $\alpha$.

\textbf{Analysis.} For the sake of simplicity, we use a Gaussian model to approximate the above distribution due to the central limit theory. Let the random variable $Y=(X-\mu)/\delta \sim \mathcal{N}(0,1)$.  We can always find a constant $c$, which satisfies
\begin{equation}
\label{eqn:erf}
\alpha = \Pr(-c\leq Y \leq c) = \text{erf}(c/\sqrt{2})	
\end{equation}
where $\text{erf}$ is the Gaussian error function. Since we require
\begingroup\makeatletter\def\f@size{9}\check@mathfonts
\begin{eqnarray*}
&&\Pr(|\widehat{n}-n|\leq \beta n) = \Pr\left((1-\beta)n\leq \widehat{n} \leq(1+\beta)n\right)\\
&&=\Pr\left((1-\beta)n\leq \frac{m}{p} \leq (1+\beta)n)\right)\\
&&=\Pr\left(\frac{(1-\beta)np-\mu}{\delta} \leq \frac{m-\mu}{\delta} \leq  \frac{(1+\beta)np-\mu}{\delta}\right)\\
&& \geq \alpha = \Pr\left(-c \leq Y \leq  c\right)
\end{eqnarray*}
\endgroup
, we can find a constant $c$ by subjecting to the below inequality:
\begingroup\makeatletter\def\f@size{9}\check@mathfonts
\begin{eqnarray*}
c &\leq & \max\{\frac{(1+\beta)np-\mu}{\delta}, -\frac{(1-\beta)np-\mu}{\delta}\}\nonumber\\
&=&\frac{(1+\beta)np-\mu}{\delta} = \frac{\beta}{1-p}
\end{eqnarray*}
\endgroup
By substituting $p=1/2^l$ into the above inequality, we obtain
\begin{equation*}
l\leq \log_2 \frac{c}{c-\beta}	
\end{equation*}
where $c=\sqrt{2}\cdot\text{erf}^{-1}(\alpha)$. Therefore, we can obtain the following theorem.

\begin{thm}
  \label{thm:estimation}
   The optimal dimension of the tash table is equal  to $ \lceil\log_2 \frac{\sqrt{2}\cdot\text{erf}^{-1}(\alpha)}{\sqrt{2}\cdot\text{erf}^{-1}(\alpha)-\beta}\rceil$, which results in an estimation error $\leq \beta$ with a probability of at least $\alpha$.
\end{thm}

\begin{figure}[t!]
  \centering
  \subfigure[Tashing once]{
  	\label{fig:missing-detection-1}
  	\includegraphics[width=0.7\linewidth]{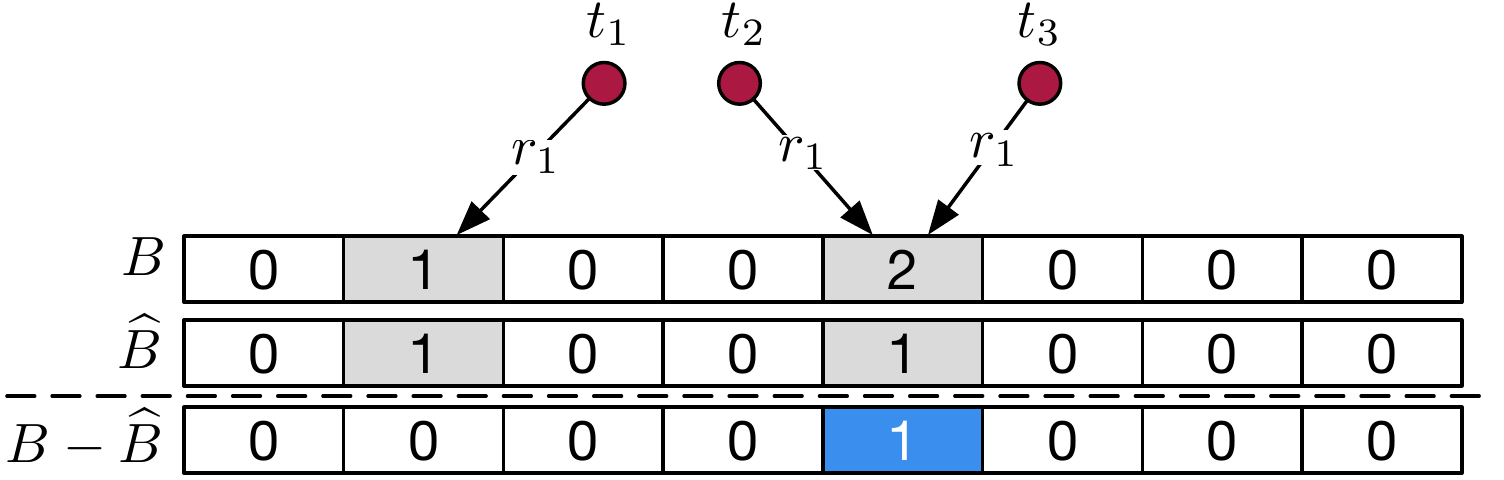}
  }
  \subfigure[Tashing twice]{
  	\label{fig:missing-detection-2}
  	\includegraphics[width=0.7\linewidth]{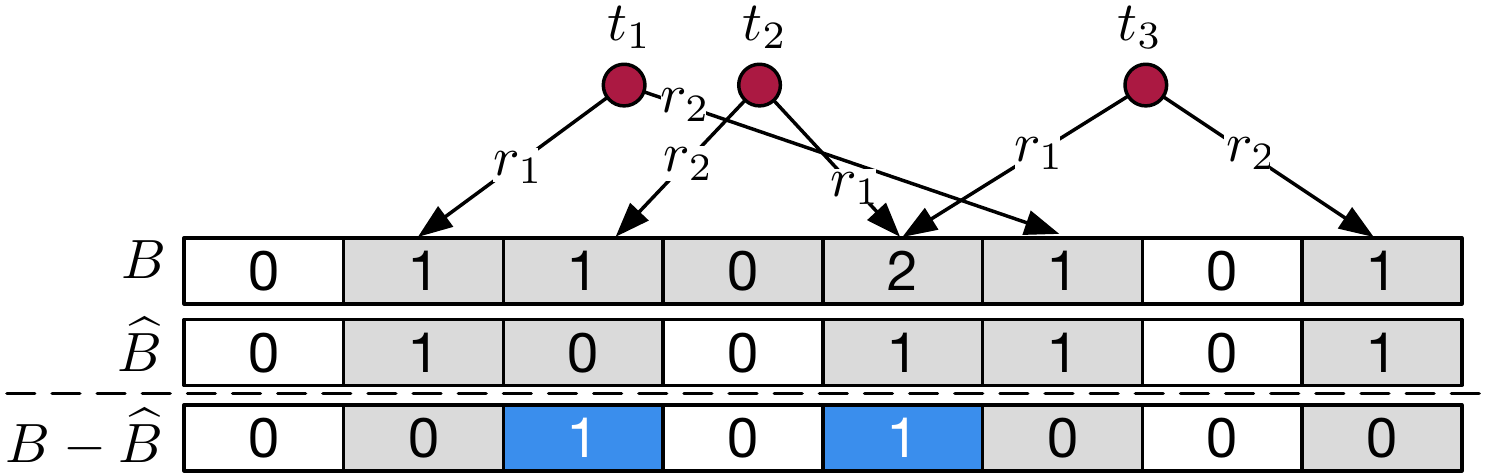}
  }
    \vspace{-0.2cm}
    \caption{An example of missing detection. \textnormal{$B$ is the intact tash table generated using the known \texttt{EPC}s while $\widehat{B}$ is an instance over the tags in the current scene. }
    \vspace{-0.3cm}
}
%===reserved for extension.
%  \caption{An example of missing detection. \textnormal{(a) There are two hash tables, $B$ and $\widehat{B}$. $B$ is the intact tash table generated using the \texttt{EPC}s stored in database while $\widehat{B}$ is generated over the tags in the current scene. (b) The missing tags can be exactly pinpointed by using the combined tash table, which are tash ORed over two or more different tash tables.}}
\end{figure}

\subsection{Usage II: Missing Tag Detection}

The purpose of missing tag detection is to quickly find out the missing tags without collecting all the tags in the scene. Such detection is very useful, especially when  thousands of tags are present. We formally define the problem of detecting missing tags in Problem~\ref{problem:missing-detection}. We assume that the \texttt{EPC}s of all the tags in a closed system are  stored in a database and known in advance. This assumption is reasonable and necessary, because it is impossible for us to tell that a tag is missing without any prior knowledge of its existence.

\begin{problem}
\label{problem:missing-detection}
How to quickly identify $m$ missing out of $n$ tags with a false positive rate of $\gamma$ at most?
\end{problem}

\textbf{Proposed detector.} The underlying idea is to compare two tash tables $B$ and $\widehat{B}$.  $B$ is an \emph{intact tash table}  created by tashing all the known \texttt{EPC}s which are stored in the database, while $\widehat{B}$ is an \emph{instance tash table} obtained from the tags in the scene. We can detect the missing tags through comparing the difference between $B$ and $\widehat{B}$. If the residual table $B-\widehat{B}$ (\ie entry-wise subtraction) equals $0$,   no missing tag event happens. Otherwise, the tags mapped into the non-zero entries of the residual table are missing.
%===== reserved for extension ==========
%Different from the prior work, which \emph{probabilistically} ensure the successful detection of missing event, our approach can \emph{deterministically} knows the occurrence of missing event since the exact number of missing tags can be acquired by adding up non-zero entries in the differential table.
 Fig.~\ref{fig:missing-detection-1} illustrates an example in which three tags, $t_1$, $t_2$ and $t_3$, are mapped into the intact tash table $B$. $\widehat{B}$ is an instance table where tag $t_2$ is missing, and thus $\widehat{B}[4]=1$. Consequently, $(B-\widehat{B})[4]=1$, we can definitely infer that one tag is missing. However,  it is impossible for us to tell which tag is missing because $t_2$ and $t_3$ are simultaneously mapped into the fourth entry.

Inspired by the Bloom filter\cite{broder2004network}, we perform $k$ tashings to identify the missing tags as follows:
\begin{equation}
	B = \mathcal{F}_l(T, r_1) || \dots || \mathcal{F}_l(T,r_k)
\end{equation}
The final $B$ after tash ORs is considered to use $k$ independent hash functions (\ie induced by $k$ random seeds) to map each tag into $B$ for $k$ times, as shown in Fig.~\ref{fig:missing-detection-2}.  The residual table of $B-\widehat{B}$ is therefore viewed as a Bloom filter which represents the missing tags. Thereafter,  to answer a query of whether a tag $t$ is missing,  we check whether all entries set by $f_l(t,r_1), \cdots$ and $f_l(t,r_k)$ in the residual table have a value of  non-zero. If the answer is yes, then tag $t$ is the missing one. Otherwise, it is not the missing tag. Fig.~\ref{fig:missing-detection-2} illustrates an example in which each tag is tashed twice. The missing tag $t_2$ can be identified because both the $2^{rd}$ and the $4^{th}$ entry in the residual table  have value of  non-zero.   Despite multiple tashings,  the query may yield a \emph{false positive}, where it suggests a tag is missing even though it is not.

\textbf{Analysis.} To lower the rate of false positive rate, it is necessary to answer two questions.

(1) \emph{How many tash functions do we need?} Given the table dimension $l$, we expect to optimize the number of tash functions. There are two competing forces: using more tash functions gives us more chance to find a zero bit for a missing tag, but using fewer tash functions increases the fraction of zero bits in the table.  After $m$ missing tags are tashed into the table, the probability that a specific bit is still $0$ is $(1-\frac{1}{L})^{km}\approx e^{-km/L}$ where $L=2^l$. Correspondingly, the probability of a false positive $p$ is given by
\begin{equation}
p = (1-e^{-km/L})^k
\end{equation}
Namely,  a missing tag falls into $k$ non-zero entries.  Lemma.~\ref{lem:optimal-tash-function-number} suggests that the optimal number of tash functions is achieved when $k=\ln 2 \cdot (L/m)$.
\begin{lem}
	\label{lem:optimal-tash-function-number}
	The false positive rate is minimized when $p=(1/2)^k$ or equivalently $k=\ln 2 \cdot (L/m)$.
\end{lem}
\begin{proof}
Please refer to \cite{broder2004network} for the proof.
\end{proof}

(2) \emph{How large tash table is necessary to represent all $m$ missing tags?}
%====== reserved for extension
%Since we totally have $n$ tags. Our representation must associate a tash table with each of the $\binom{n}{m}$ possible sets. Recall that any tash table $B-\widehat{B}$ is used to represent $m$ missing tags without \emph{false negatives}, but it many also accept $(n-m)\beta$ other elements of the entire tag sets while maintaining a false positive rate of at most $\beta$. In other words, $B-\widehat{B}$ accepts at most $m+(n-m)\beta$ tags. Simply, we only consider two cases for each entry (\ie either zero or non-zero), then we have $2^L$ distinct result tables that must represent all $\binom{n}{m}$ sets. Hence, we must have
%\begingroup\makeatletter\def\f@size{9}\check@mathfonts
%\begin{equation}
%	2^L\binom{m+(n-m)\beta}{m}\geq \binom{n}{m},
%\end{equation}
%\endgroup
%or
%\begingroup\makeatletter\def\f@size{9}\check@mathfonts
%\begin{equation}
%\label{eqn:missing-tash-size}
%	L\geq \log_2 \frac{\binom{n}{m}}{\binom{m+(n-m)\beta}{m}}\approx \log_2 \frac{\binom{n}{m}}{\binom{n\beta}{m}} \geq  m\log_2(1/\beta).
%\end{equation}
%\endgroup
%The approximation above is suitable when $m$ is small compared to $n\beta$. Recall that the optimal value for $k$ is $L\ln 2/m$. The above inequality gives a \emph{lower bound} about the size of tash table.
Recall that the false positive rate achieves minimum when $p=(1/2)^k$. Let $p\leq \gamma$. After some algebraic manipulation, we find
\begin{equation}\small
	L\geq \frac{m\log_2 (1/\gamma)}{\ln2} = m\log_2 e \cdot \log_2(1/\gamma) = 1.44m \log_2(1/\gamma)
\end{equation}
%Thus, the space-wise tash table are within a factor of $\log_2 e \approx 1.44$ of the lower bound.

\begin{figure}[t!]
  \centering
  \includegraphics[width=0.7\linewidth]{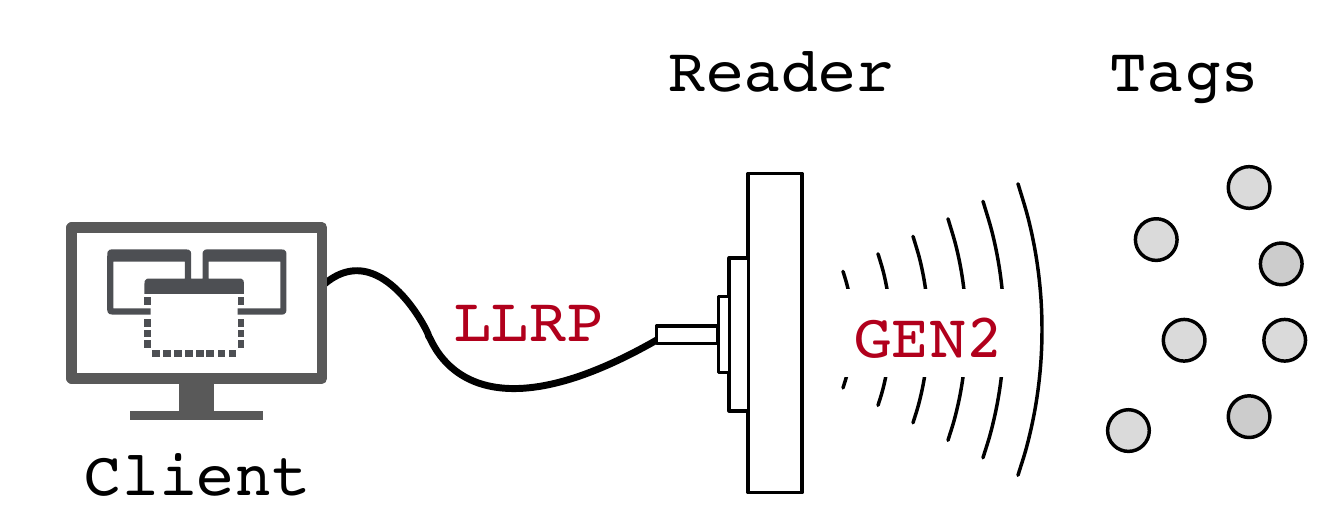}
   \caption{Gen2 vs. LLRP. \textnormal{Gen2 is the air protocol between a reader and tags while LLRP is the driver protocol between a client computer and a reader. Our framework leverages LLRP to manipulate a reader to broadcast  Gen2 commands that we need. }}
  \label{fig:llrp}
\end{figure}

\noindent Finally, putting the above conclusions together, we have the subsequent theorem.
\begin{thm}
\label{theorem:missing-detection}
Setting the table dimension to $\lceil \log_2(1.44m$ $\log_2$ $(1/\gamma)) \rceil$ and using $\lceil\ln 2 \cdot (2^l/m) \rceil$ random seeds allow the false positive rate of identifying $m$ missing tags lower than a given tolerance $\gamma$.
\end{thm}
\noindent
%An extremely small $\gamma$ may introduce an unbelievable large $l$. In this situation,  we should consider the time cost consumed on per-tag collection in theory to make a better decision.
%, and collect the cost-effective one.

%\begin{figure}[t!]
%  \centering
%  \includegraphics[width=0.7\linewidth]{llrp-config}
%  \vspace{-0.3cm}
%   \caption{LLRP RO specification. \textnormal{The specification of the XML file defines various parameters that are required for selection commands.}}
%  \label{fig:llrp-config}
%\end{figure}
%
%
%\begin{figure}[t!]
%  \centering
%  \includegraphics[width=0.7\linewidth]{interface}
%   \caption{Tash framework interface. \textnormal{The primary interfaces provides by tash framework, which is developed by using Java language and LLRP Toolkit.}}
%  \label{fig:tash-interface}
%\end{figure}

\begin{figure}[t!]
  \centering
  \subfigure[LLRP RO specification. ]{
  	\label{fig:llrp-config}
  	\includegraphics[width=0.7\linewidth]{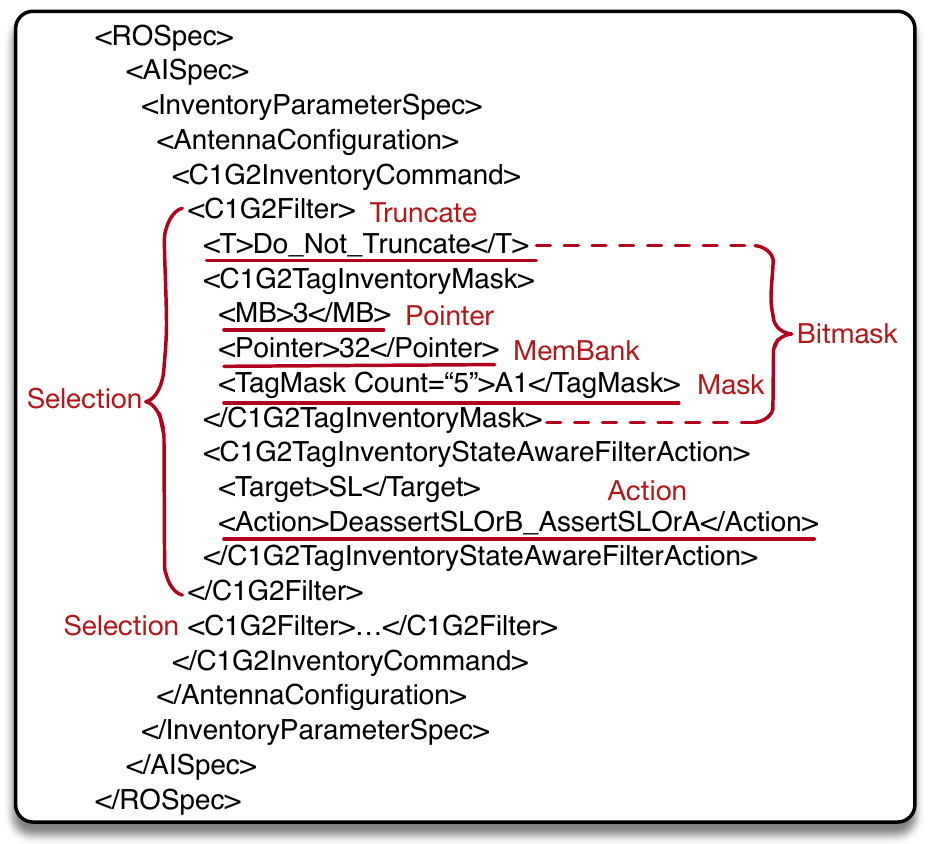}
  }
  \subfigure[Tash framework interface. ]{
  	\label{fig:tash-interface}
  	\includegraphics[width=0.7\linewidth]{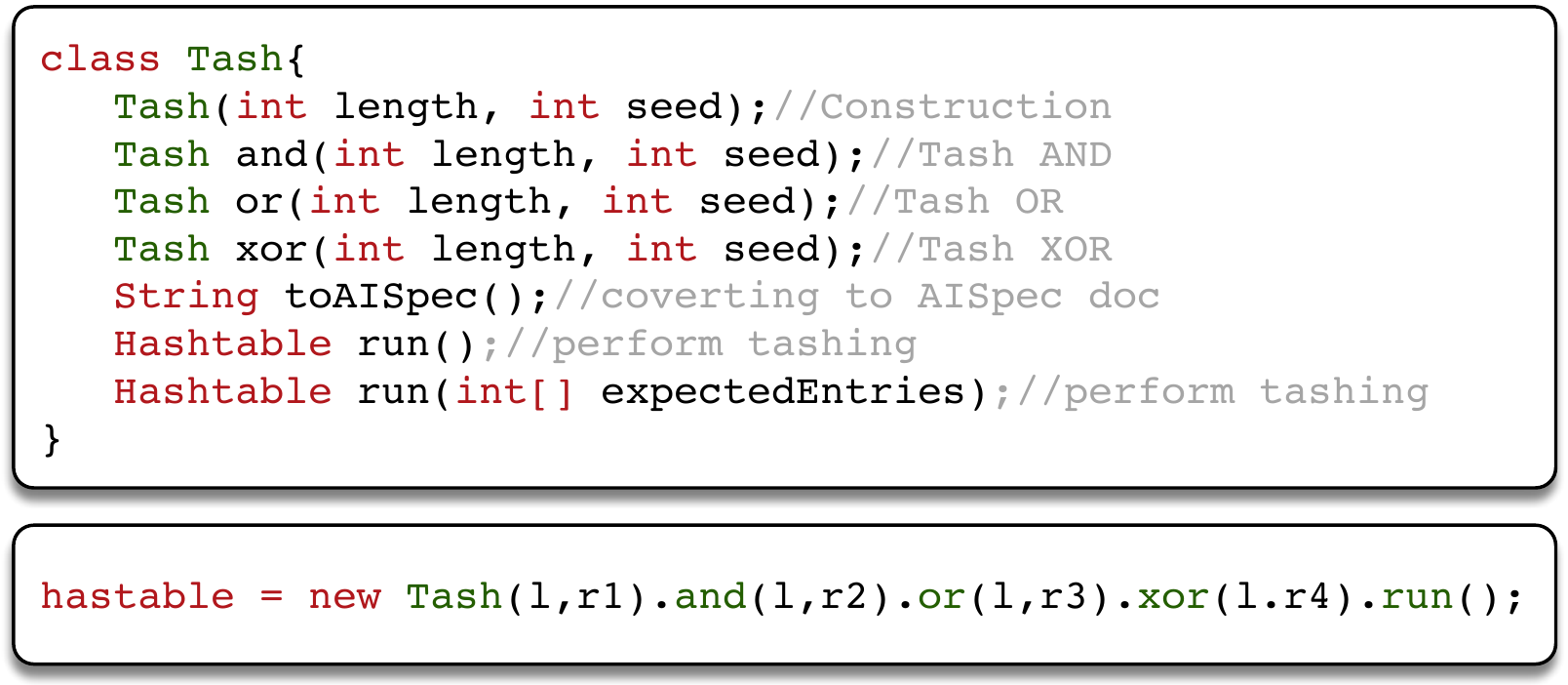}
  }
    \caption{Tash implementation. \textnormal{(a) The specification of the XML file defines various parameters that are required for selection commands. (b) The primary interfaces provides by tash framework, which is developed by using Java language and LLRP Toolkit.}}
      \vspace{-0.3cm}
\end{figure}

\section{Tash Implementation}
\label{section:implementation}

Our implementation involves two kinds of protocols: UHF Gen2 air interface protocol (Gen2) and Low Level Reader Protocol (LLRP). \underline{As shown in Fig.~\ref{fig:llrp}},  Gen2 protocol defines the physical and logical interaction between readers and passive tags, while LLRP allows a client computer to control a reader.  Each client computer connects one ore more RFID readers via Ethernet cables.  LLRP is the \emph{driver program (or driver protocol)} for  Gen2 readers. We leverage LLRP to manipulate a reader to broadcast Gen2 commands that we need.  Notice that we do not need particularly implement Gen2 protocol, which has been implemented in the COTS RFID devices that we are using.  Specifically, LLRP specifies two types of operations: reader operation (RO) and access operation (AO). Both operations are represented in  XML document form and transported to a reader through TCP/IP.

\textbf{Reader operation.} RO defines the inventory parameters specified in the Gen2 protocol, such as bitmask, antenna power, and frequency. Fig.~\ref{fig:llrp-config} shows a simplified instance of an \texttt{ROSpec}.  An \texttt{ROSpec} is composed of at least one \texttt{AISpec}. Each \texttt{AISpec} is used for an antenna setting. An \texttt{AISpec} consists of more than one \texttt{C1G2Filter}s. The filter functions as a bitmask. We can set multiple selection commands by adding multiple \texttt{C1G2Filter}s.

\textbf{Access operation.} AO defines the access parameters for writing or reading data  to and from a tag.  We leverage the \texttt{C1G2Write}  inside an \texttt{AOSpec} to write the hash value of the \texttt{EPC} into a user-defined memory bank. As the \texttt{EPC}s are highly related to the products the tags attached, the writing of hash values should be accomplished by the product manufactures or administrators.  There is almost no overhead to write data into \texttt{MemBank-3} since it is allowed to write a batch of tags simultaneously using \texttt{Write} commands specified in one \texttt{AOSpec}, without physically changing tags' positions.

\textbf{Tash framework}. Our framework is developed by using Java language and the LLRP Toolkit\cite{llrp-toolkit}, which is an open-source library for handling \texttt{ROSpec} and \texttt{AOSpec}.  Fig.~\ref{fig:tash-interface} shows the primary interfaces provided by the tash framework. The class \texttt{Tash} makes the first selection through its construction method and allows the calls of three operators to be chained together in a single statement.
The method \texttt{toAISpec} converts a \texttt{Tash} object or  a chain of  \texttt{Tash} objects into an \texttt{AISpec}. The entry-inventories are physically executed in the connected reader when the method \texttt{run} is invoked.  This method allows users to make selective entry-inventories by passing an index array. For example, the operation $\mathcal{F}_l(T, r_1) \oplus \mathcal{F}_l(T, r_2)  || $ $\mathcal{F}_l(T,r_3) \otimes \mathcal{F}_l(T,r_4)$ can be coded in a manner similar to that shown at the bottom of Fig.~\ref{fig:tash-interface}.

\begin{table*}[!t]
  \centering
  \footnotesize
  \caption{Summary of Gen2-compatibility on tag.}
  \label{tab:tag-compatibility}
  \vspace{-0.3cm}
  \begin{tabular}{|l|c|c|c|c|c|c|c|c|c|c|c|c|c|c|c|c|c|c|}
   \hline
   &  \multicolumn{9}{|c|}{\textbf{ImpinJ Monza}} & \multicolumn{9}{|c|}{\textbf{Alien ALN}} \\
    \hline
    Commands & 5 & D & E & QT & X-2K& X-8K &R6 & R6-P & R6-C &
    9840 & 9830 &9662 &9610&9726&9820&9715&9716&9629 \\
     \hline
     \texttt{MemBank1} (bits) &128 & 128 &496 & 128 & 128& 128 &96 & 128/96 & 96 &
    128 & 128 & 480 &96-480&128&128&128&128&96 \\
    \hline
     \texttt{MemBank3} (bits) &32 & 32 &128 & 512 & 2176& 8192 &\texttimes & 32/64 & 32 &
    128 & 128 & 512 &512&128&128&128&128&512 \\
    \hline
 \texttt{Write} cmd& \checkmark  & \checkmark  & \checkmark  & \checkmark  & \checkmark & \checkmark  &\texttimes  & \checkmark  & \checkmark  &
    \checkmark  & \checkmark  &\checkmark  &\checkmark &\checkmark &\checkmark &\checkmark &\checkmark &\checkmark  \\
 \hline
    \texttt{Select} cmd & \checkmark   &\checkmark   & \checkmark   & \checkmark   & \checkmark  & \checkmark  &\checkmark   & \checkmark  & \checkmark   &
    \checkmark   & \checkmark   &\checkmark   &\checkmark  &\checkmark  &\checkmark  &\checkmark  &\checkmark  &\checkmark   \\
    \hline
    \texttt{Truncate} cmd&  --   & --  &  --  &  --  &  -- &  -- & --  &  -- &  --  &
     --  &  --  & --  & -- & -- & -- & -- &--  & -- \\
       \hline
  \end{tabular}
\end{table*}

\section{Microbenchmark}
\label{section:microbenchmark}

We start with a few experiments that provide insight to our hash primitives.

\subsection{Experimental Setup}
\begin{figure}[t!]
  \centering
  \includegraphics[width=0.6\linewidth]{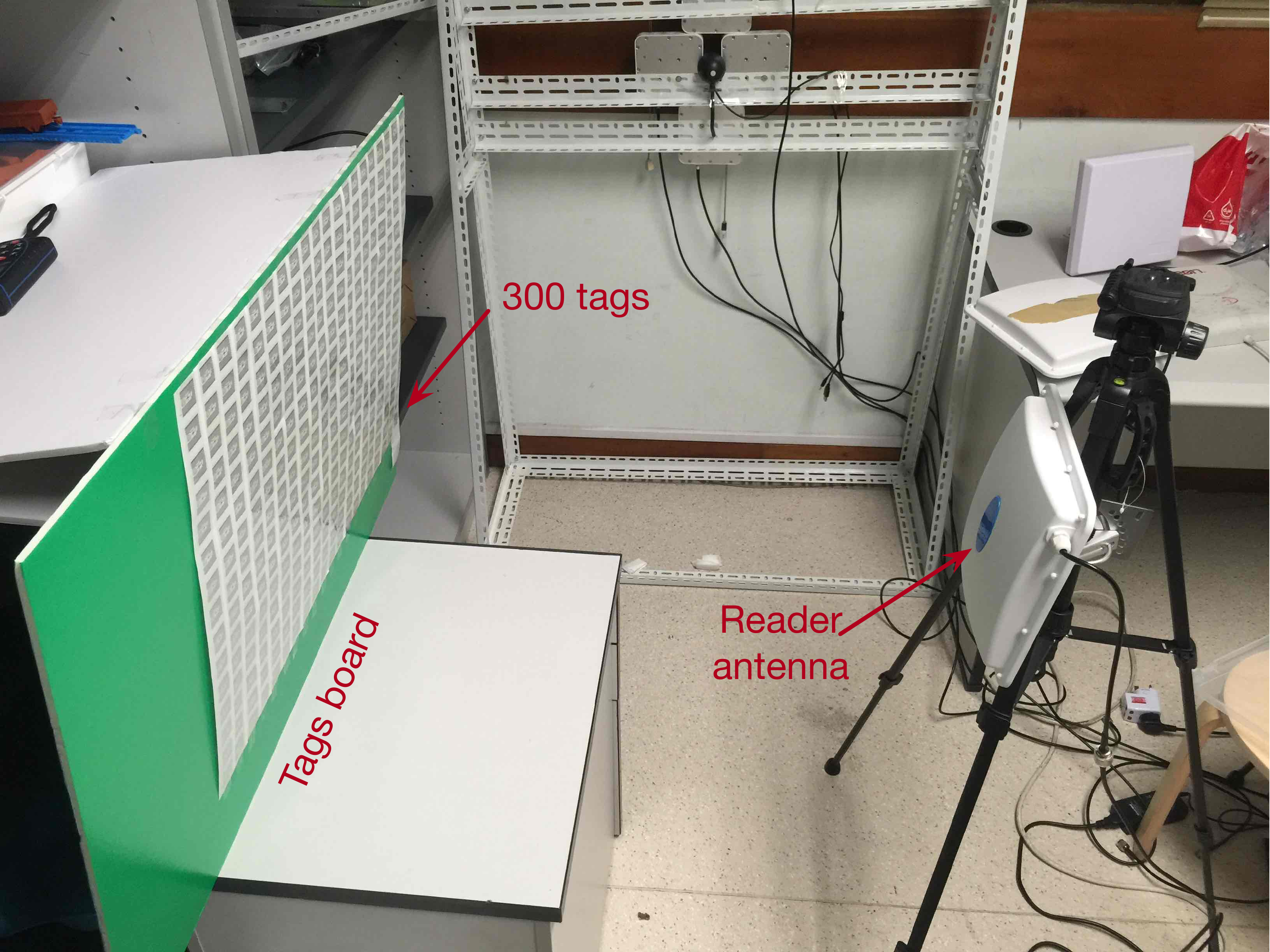}
   \caption{Testbed in our laboratory. \textnormal{Total $300$ tags are attached on a board and covered by a directional reader antenna.}}
  \label{fig:lab-scene}
  \vspace{-0.3cm}
\end{figure}

We evaluate the framework using COTS UHF readers and tags. We use a total of $3$ models of ImpinJ readers (R220, R420 and R680), each of which is connected to a $900$MHz and $8$dB gain directional antenna.  In order to better understand the feasibility and effectiveness of \oursystem in practice, we test a total of $3,000$ COTS tags with different models. We divide these tags into $10$ groups of $300$ tags each. The tags of each group are densely attached to a plastic board which is placed in front of a reader antenna. As shown in Fig.~\ref{fig:lab-scene}, three hundreds is the maximum number of tags that can be covered by one directional antenna in our laboratory. We store the $3,000$ \texttt{EPC} numbers in our database as the ground truth.  The 128-bit MD5 is employed as the common hash function to generate the hash values of \texttt{EPC}s. The experiments with the same settings are repeated across the $10$ groups, and the average result is reported.

\subsection{Compatibility Investigation}

First, we investigate the compatibility of Gen2 across $10$ different  types of readers and $18$ different types of  tags in terms of  the functions or commands that \oursystem requires. The readers and tags may come from different manufacturers but work together in practice. These investigated products  are all publicly claimed to be completely Gen2-compatible.

\textbf{Reader compatibility.} We investigate the R220, R420, and R680 models from ImpinJ\cite{impinj}, the Mercury6, Sargas  and M6e models from ThingMagic\cite{thingmagic}, as well as the ALR-F800, 9900+, 9680 and 9650 models from Alien\cite{alien}. We perform the investigation through real tests for the first three models of readers (\ie the ImpinJ series), and investigate the other readers through their data sheets or manuals (because we are limited by the lack of hardware). The Gen2-compatibility of readers is briefly summarized in Table.~\ref{tab:reader-compatibility}. Consequently, we have the subsequent findings. ($1$) All the readers do support \texttt{Write/Read} command, which Tash uses for writing or reading hash values of \texttt{EPC} numbers. ($2$) All the readers do support the \texttt{Select} command, which Tash uses for the selective reading. ($3$) However, our practical tests suggest that none model of the ImpinJ series supports the  \texttt{Truncate} command, which Tash uses to hear the one-bit presence signal. The serviceability  of other readers is not clearly indicated in the manuals of those readers.  ($4$) The Gen2 protocol does not specify how many \texttt{C1G2Filter}s and \texttt{AISpec}s  that a reader should support. Our practical tests suggest that the ImpinJ series supports $4$ \texttt{C1G2Filter}s and $16$ \texttt{AISpec}s, which means that we can only use a maximum of four tash operators each time.

\begin{table}[!b]
  \centering
  \small
  \caption{Summary of Gen2-compatibility on reader}
    \vspace{-0.3cm}
  \label{tab:reader-compatibility}
  \begin{tabular}{|l|c|c|c|}
    \hline
    \textbf{Commands or functions} & \textbf{ImpinJ} & \textbf{ThingMagic} & \textbf{Alien} \\
    \hline
    \texttt{Write}/\texttt{Read}& \checkmark & \checkmark & \checkmark\\
    \hline
    \texttt{Select} & \checkmark & \checkmark & \checkmark \\
    \hline
    \texttt{Truncate} & \texttimes & -- & --\\
    \hline
    Max No. of \texttt{C1G2Filter}s & $4$ & -- & --\\
    \hline
    Max No. of \texttt{AISpec}s & $16$ & -- & -- \\
    \hline
  \end{tabular}
\end{table}

\textbf{Tag compatibility}.  We investigate $9$ chip models from ImpinJ Monza series and $9$ additional models from Alien ALN series. The majority of tags on the market contain these $18$ models of chips and customized antennas. Table.~\ref{tab:tag-compatibility} summarizes the result of our investigation, from which we have the subsequent findings. (1) Tags reserve $96\sim 480$ bits of memory for storing \texttt{EPC} numbers, among which the size of $96$ bits has become the de facto standard.  (2) \oursystem requires \texttt{MemBank-3} to store the hash values. The results of the investigation show that almost all tags allow to write to and read from the third memory bank, with an exception of ImpinJ Monza R6, which does not have the user-defined memory. The size of the third memory bank fluctuates around $32\sim 512$ bits. The de facto standard has become $128$ bits. (4) All tags are claimed to support the \texttt{Truncate} command according to their public data sheets. However, we have no idea about their real serviceability   due to the lack of \texttt{Truncate}-supportable reader available for practical tests. In our future work, we plan to utilize USRP for further tests.

\textbf{Summary.} Despite positive and public claims, our investigation shows that current COTS RFID devices, regardless of readers or tags and models, have some defects in their compatibility with Gen2, especially with regard to  \texttt{Truncate}.  The reason, we may infer, is that these commands are seldom used in practice and therefore never receive attention from manufacturers. The partial compatibility of such devices cannot fully achieve the performance \oursystem brings. Even so, we are obliged to make the claim, again, that our design strictly follows the Gen2 protocol. We hope this work can encourage manufacturers to upgrade their products (\eg reader firmware) to achieve full compatibility.

%\begin{figure*}[!t]
%  \centering
%    \begin{minipage}{0.33\linewidth}
%    \centering
%    \includegraphics[width=\linewidth,height=3.6cm]{tash-randomness}
%    \vspace{-0.5cm}
%    \caption{Randomness of tash}
%    \label{fig:random-test-cdf}
%  \end{minipage}
%      \begin{minipage}{0.33\linewidth}
%    \centering
%    \includegraphics[width=\linewidth,height=3.6cm]{hash-table-balance}
%        \vspace{-0.5cm}
%    \caption{Balance of tash table}
%    \label{fig:hash-table-balance}
%  \end{minipage}
%    \begin{minipage}{0.33\linewidth}
%    \centering
%    \includegraphics[width=\linewidth,height=3.6cm]{hash-table-time}
%        \vspace{-0.5cm}
%    \caption{Speed of tash table function}
%    \label{fig:tash-table-speed}
%  \end{minipage}
%\end{figure*}
%
%\begin{figure*}[!t]
%  \centering
%    \begin{minipage}{0.33\linewidth}
%    \centering
%    \includegraphics[width=\linewidth,height=3.6cm]{hash-operator-or}
%    \vspace{-0.5cm}
%    \caption{Performance of tash OR}
%    \label{fig:hash-operator-or}
%    \vspace{-0.2cm}
%  \end{minipage}
%      \begin{minipage}{0.33\linewidth}
%    \centering
%    \includegraphics[width=\linewidth,height=3.6cm]{usage-estimation-real}
%    \vspace{-0.5cm}
%    \caption{Estimation via testbed}
%    \label{fig:usage-estimation-real}
%    \vspace{-0.2cm}
%  \end{minipage}
%    \begin{minipage}{0.33\linewidth}
%    \centering
%    \includegraphics[width=\linewidth,height=3.6cm]{usage-estimation-time}
%    \vspace{-0.5cm}
%    \caption{Estimation via simulation}
%    \label{fig:usage-estimation-time}
%    \vspace{-0.2cm}
%  \end{minipage}
%
%\end{figure*}

\subsection{Tash Function}

\begin{figure}[t!]
  \centering
  \subfigure[Distributions of percents of `0' and `1']{
  	\label{fig:random-percent}
  	\includegraphics[width=0.8\linewidth]{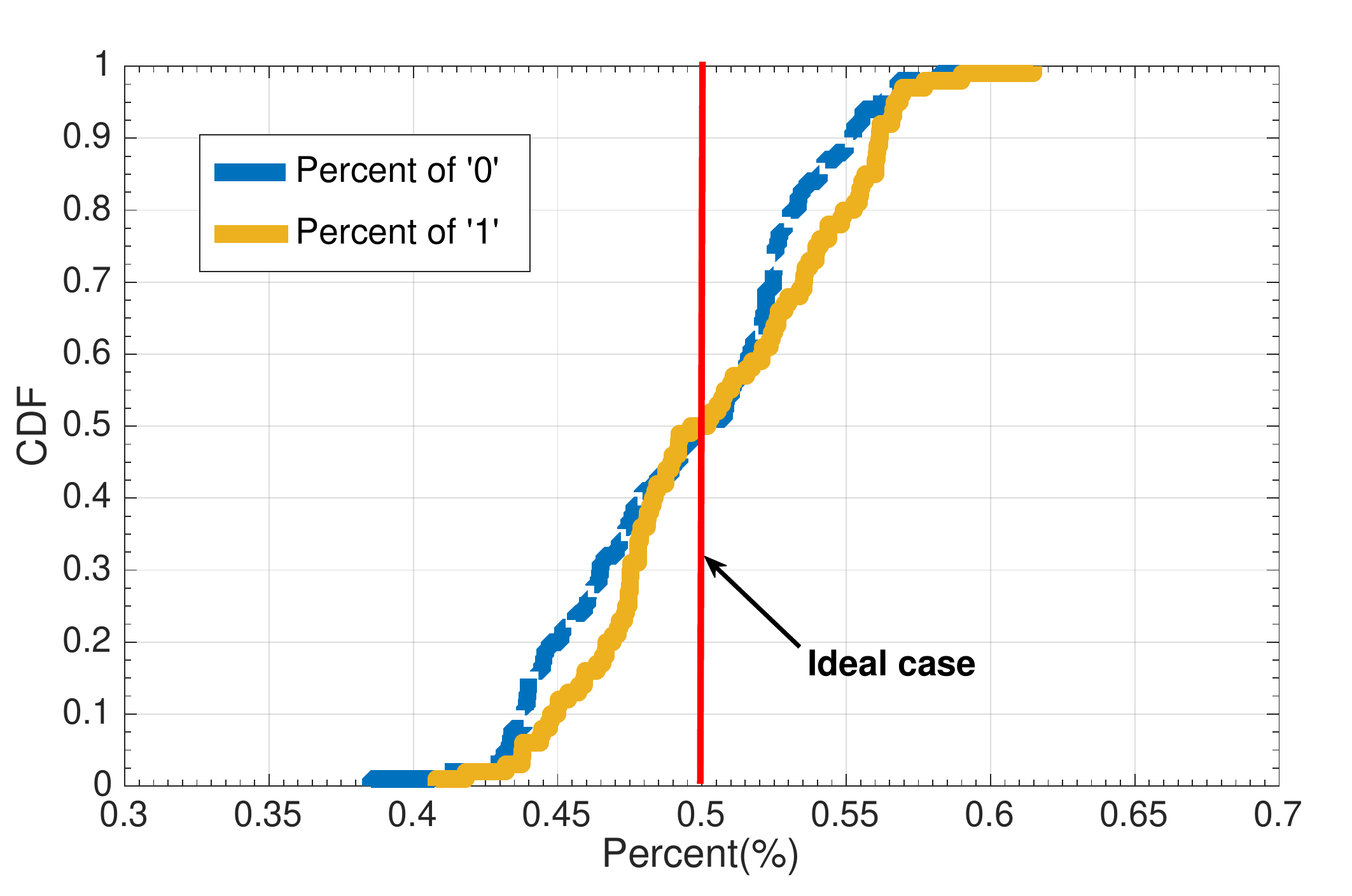}
  }
  \subfigure[Results of random test]{
  	\label{fig:random-test}
  	\includegraphics[width=0.8\linewidth]{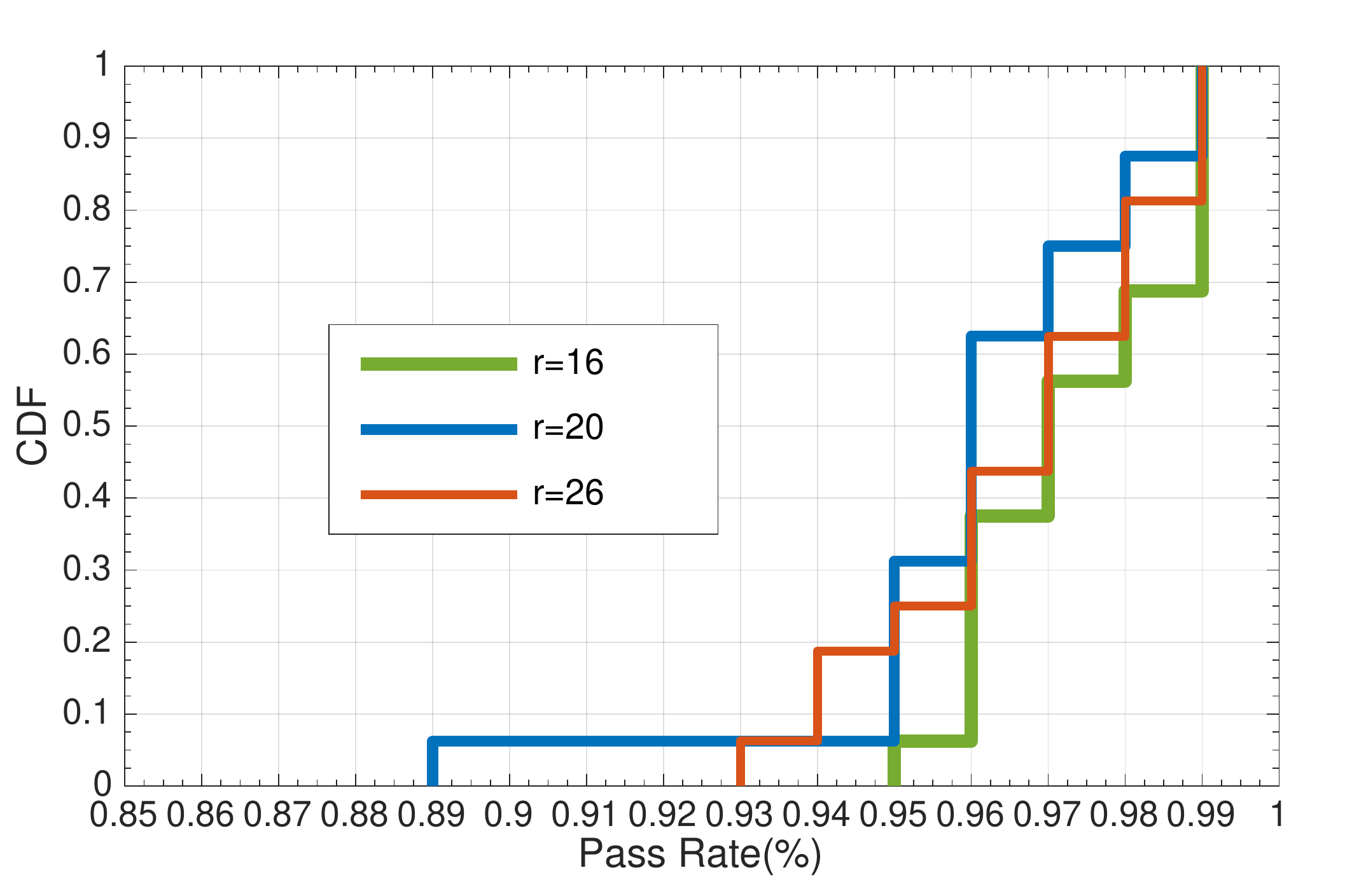}
  }
  \vspace{-0.2cm}
    \caption{Evaluation of tash function. \textnormal{(a) shows the CDF of percents of `0' and `1' appearing in the tash values. (b) shows the CDF of pass rates of random ness tests.}}
    \vspace{-0.3cm}
\end{figure}

\begin{figure}[t!]
  \centering
  \subfigure[Balance]{
  	\label{fig:hash-table-balance}
  	\includegraphics[width=0.8\linewidth]{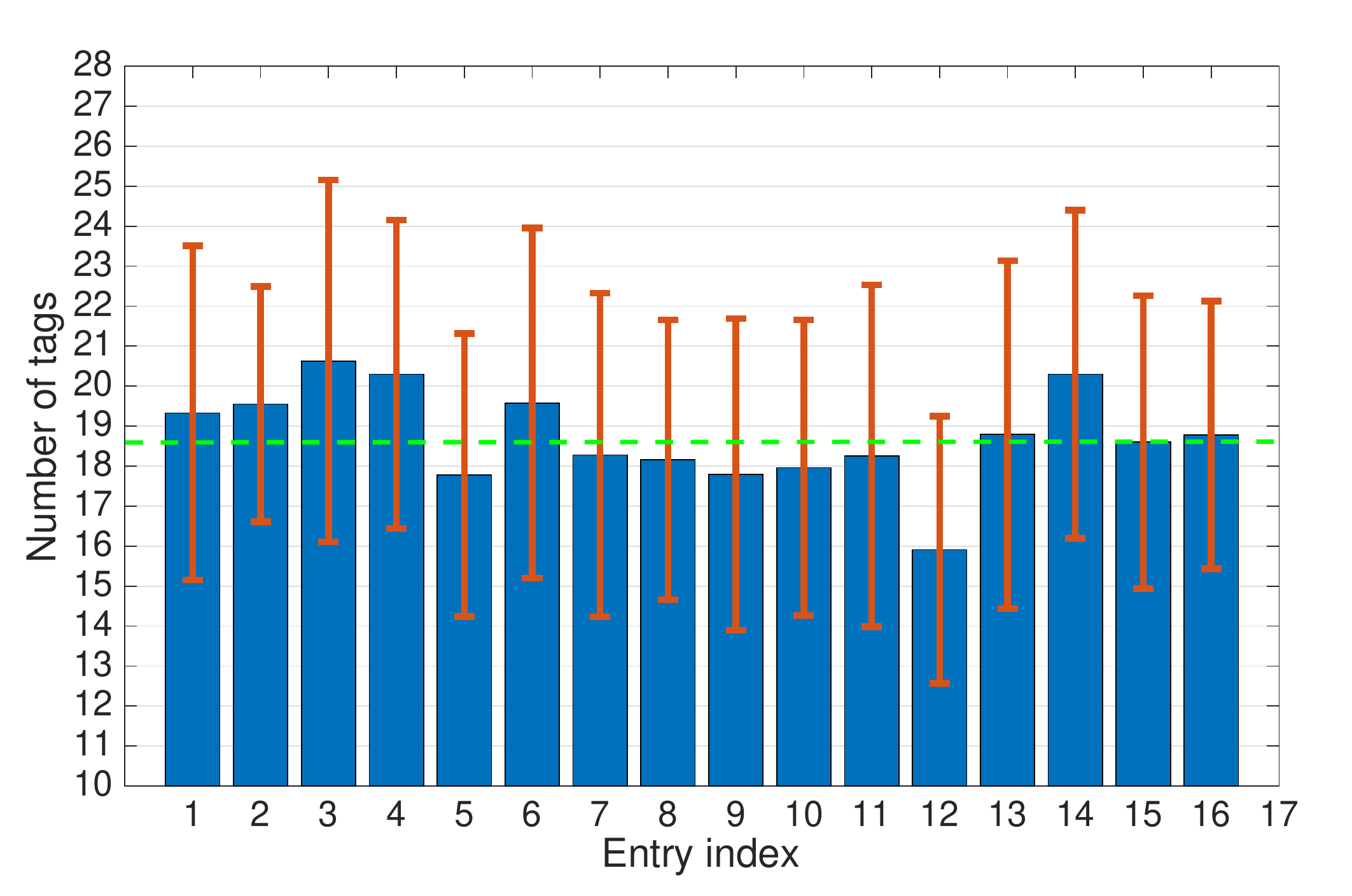}
  }
  \subfigure[Gathering time]{
  	\label{fig:tash-table-speed}
  	\includegraphics[width=0.8\linewidth]{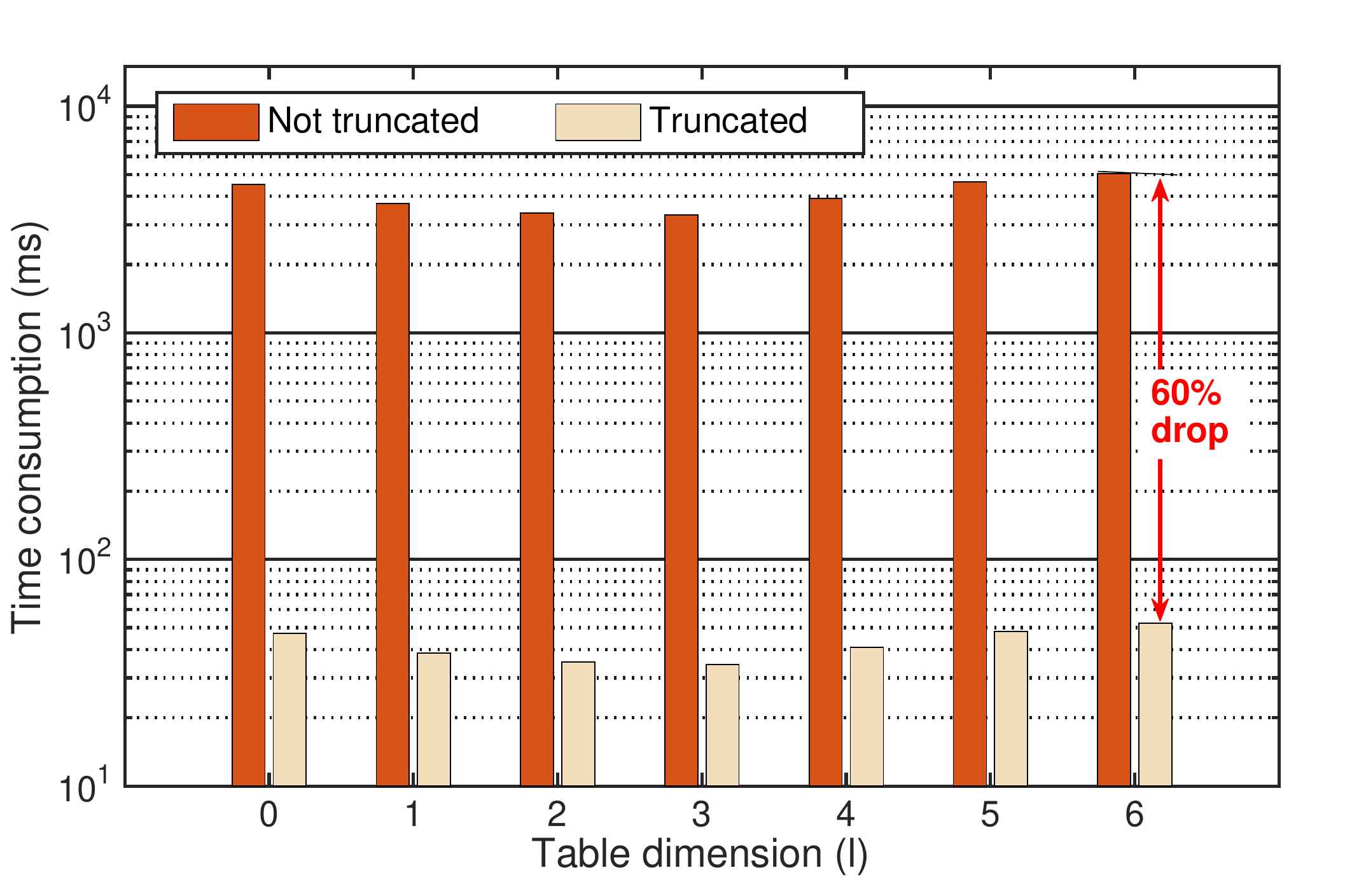}
  }
  \vspace{-0.2cm}
    \caption{Evaluation of tash table function. \textnormal{(a) shows the balance of a $4$-bit tash table across $300$ tags using $100$ different random seed. (b) shows the time consumption on gathering $6$ hash tables with different dimensions.}}
    \vspace{-0.3cm}
\end{figure}

Second, we evaluate the tash function with respect to the randomness and the accessibility.

\textbf{Randomness}. Randomness is the most important metric for a hash function. It requires that the outputs of a hash function must be uniformly distributed.  To validate the randomness of the tash function, we collect $99,886$ real \texttt{EPC} numbers from our partner (\ie an international logistics company), which introduced RFID technology for sorting tasks five years ago. Each \texttt{EPC} number has a length of $96$ bits and encodes the basic information about the package, such as sources, destinations, serial numbers, and so on. We employ the $128$-bit MD5 to create the hash values of these \texttt{EPC}s. As the minimum size of the \texttt{MemBank-3} is $32$ bits (see Table.~\ref{tab:reader-compatibility}), we choose to use only the first $32$ bits for our tests. We traverse $r$ and $l$ from $0\sim31$ and $1\sim 32-r$ respectively. For each pair of $r$ and $l$, we obtain $99,886$ tash values over all the \texttt{EPC}s. Across these tash values, we further conduct the following two analysis: (1) We merge $100$ tash values, which are randomly selected from the above results, into a long bit string. We then calculate the percents of `0' and `1' emerged in that bit string. This operation is repeated for $100$ times. Finally, totally $100$ pairs of percents are obtained. Their CDFs are plotted in  Fig.~\ref{fig:random-percent}. Ideally, each bit has a equal probability of $0.5$ to be zero or one if a hash function makes a good randomicity. From the figure, we can figure out that the percents distributed between $0.4$ and $0.6$. In particular,  percents of `0' and `1' have means of $0.49$ and $0.50$ with standard deviations of $0.043$ and $0.044$ respectively.  (2)  We shuffle these values into $100$ groups, and employ the $\chi^2$-test with a significance level of $0.05$ to test each group's goodness-of-fits of the uniform distribution (\ie passed or failed). Then, we finally calculate the pass rate for a pair of setting. In this manner, we totally obtain $496$ pass rates. More than $60\%$ of the pass rates are over than $0.95$. In particular, three sets of the results with $r=16$, $20$ and $26$ and a variable $l$, are selected to show in Fig.~\ref{fig:random-test}. We find that $90\%$ of the pass rates exceed $0.95$ for the three cases, and their median pass rates are around $0.97$. Thus, the two above statistical results suggest that our tash function has a very good quality of randomness.

%\begin{figure}[t!]
%  \centering
%  \includegraphics[width=0.7\linewidth]{tash-randomness}
%   \caption{Evaluation of tash function. \textnormal{dd}}
%  \ \label{fig:random-test-cdf}
%\end{figure}

\textbf{Accessibility}. Accessibility  refers to the ability to get access to a tash value from a tag. As aforementioned, we have two ways to acquire the tash values. The first way is to use the \texttt{Read} command. The second way is to indirectly access a tash value through a selective reading. We choose the second method since it is the basis of our design. Specifically, we perform a selective reading to determine whether the tags are collected as expected, when given random inputs and a possible tash value. We intensively and continuously perform such readings across the $10\times 300$ tags using three 4-port ImpinJ readers for three rounds of $24$ hours in a relatively isolated environment (\eg an empty room without disturbance). Surprisingly, we find all the reading results faithfully conform to our benchmarks without any exceptions. This shows that the selective reading is well supported by the manufactures and is both stable and reliable.
%Meanwhile, it also shows that obtaining tash values through the selective readings is completely feasible.

\subsection{Tash Table Function}

Third, we evaluate the performance of the tash table function in terms of its balance and gathering speed.

\textbf{Balance}. A good hash table function will equally assign each key to a bucket. We expect the output tash table to be as balanced as possible.  To show this feature, we generate $100$ different $4$-bit tash tables (\ie each includes $16$ entries) across $300$ tags using $100$ different random seeds. If the tash table is well balanced, the expected number of each entry should be very close to $300/16=18.7$. Fig.~\ref{fig:hash-table-balance} shows the mean number of tags in each entry as well as their standard deviations. The average number across $16$ entries equals $18.75$, which is very close to the expected  theoretical value. The average standard deviation equals $0.44$. Thus, the good randomness quality of tash functions results in output tash tables being well balanced.

\begin{figure}[t!]
  \centering
  \includegraphics[width=0.8\linewidth]{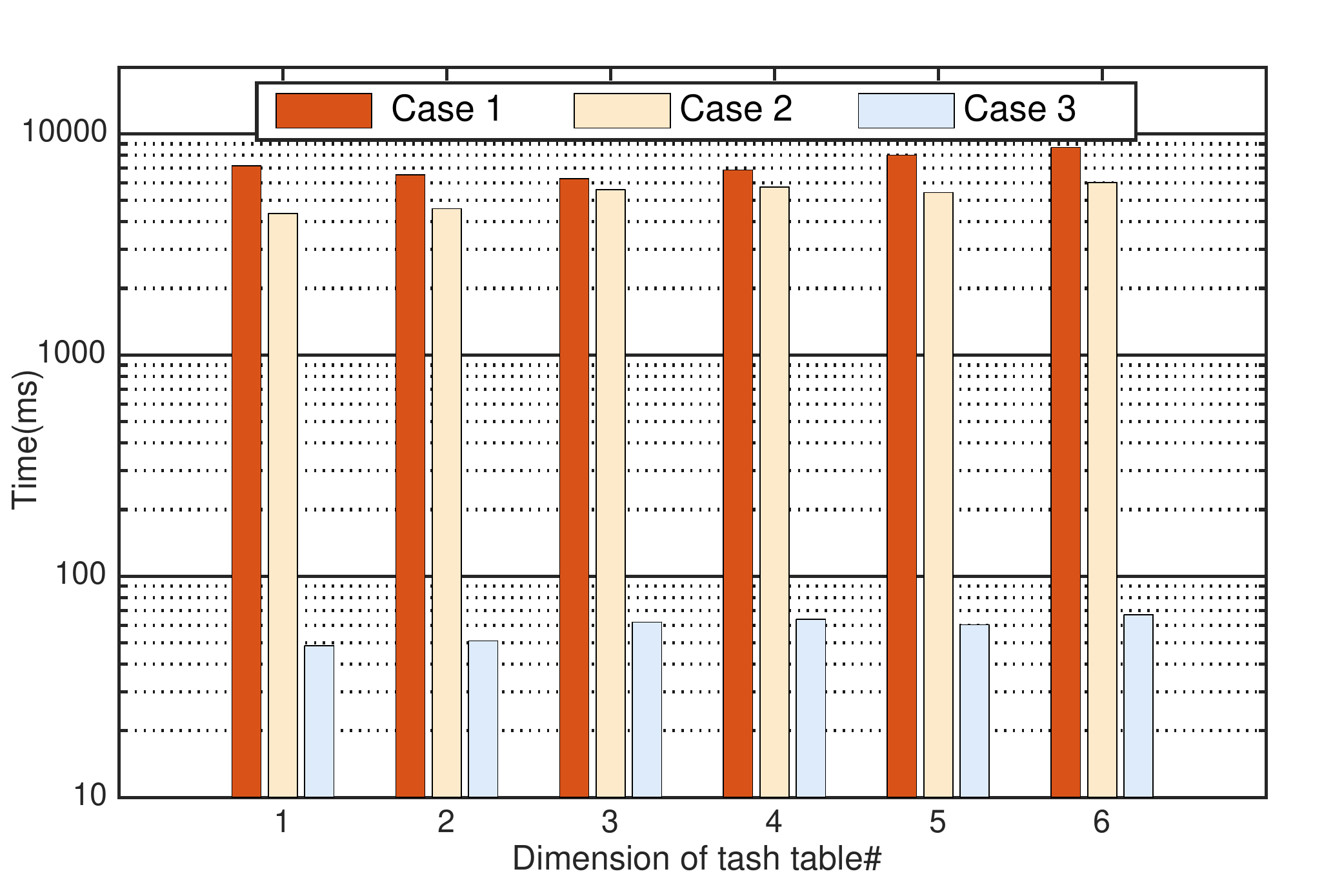}
  \vspace{-0.2cm}
    \caption{Performance of tash OR. \textnormal{In Case 1, two tash tables are conducted OR operation on application layer without truncating replies.  In Case 2 and Case 3  two tables are conducted Tash OR on tags without/with  truncating replies.}}
    \label{fig:hash-operator-or}
        \vspace{-0.3cm}
\end{figure}

\textbf{Gathering speed}. We then consider the time consumption of gathering a tash table. Fixing the random seed, we vary the table dimension $l$ from $0$ to $6$. We then measure the  time taken on gathering a tash table with the deployed $300$ tags.  Fig.~\ref{fig:tash-table-speed} shows the resulting time as a function of the table dimension. From the immediately above-mentioned  figure,  we can observe the subsequent findings.  ($1$)  When $l=0$ without truncating reply, the result is equivalent to collecting $300$ complete \texttt{EPC}s of all the tags. Such time consumption (\ie $4,524ms$) is viewed as our baseline.  ($2$) By contrast, when $l>0$ without truncating a reply, the collection amounts to dividing all the tags into $2^l$ groups ``equally'' and then collecting each group independently.  In this manner, when $l\leq 4$, such ``divide and conquer'' approach is better than ``one time deal'', \ie a drop in overhead of about $10\%$. The Gen2 reader uses a Q-adaptive algorithm for the anti-collision. This algorithm is able to adaptively learn the best frame length from the collision history. Due to the division,  a smaller number of tags can make reader's learning relatively quicker and improve the overall performance.  ($3$) However, when $l>4$, the performance of ``divide and conquer'' approach starts to deteriorate. The ImpinJ reader  supports $16$ \texttt{AISpec}s at most (see Table.~\ref{tab:reader-compatibility}). We have to re-send another \texttt{ROSpec} for the remaining selective readings when the number of entry-inventory is above $16$ (\ie $l>4$), which introduces additional time consumption. ($4$) We then consider the case where the reply is truncated to a one-bit presence signal as assumed by HEPs. Due to the defects of ImpinJ readers in the implementation of the \texttt{Truncate} command, we cannot measure the actual time spent on collecting truncated \texttt{EPC}s. We can only utilize the least-square algorithm to estimate the transmission  time for a one-bit presence signal. Our fitting results show that truncating reply would introduce about $60\%$ drop of the overhead at least.

\subsection{Tash Operators}

\begin{figure}[t!]
  \centering
  \subfigure[Testbed]{
  	\label{fig:usage-estimation-real}
  	\includegraphics[width=0.8\linewidth]{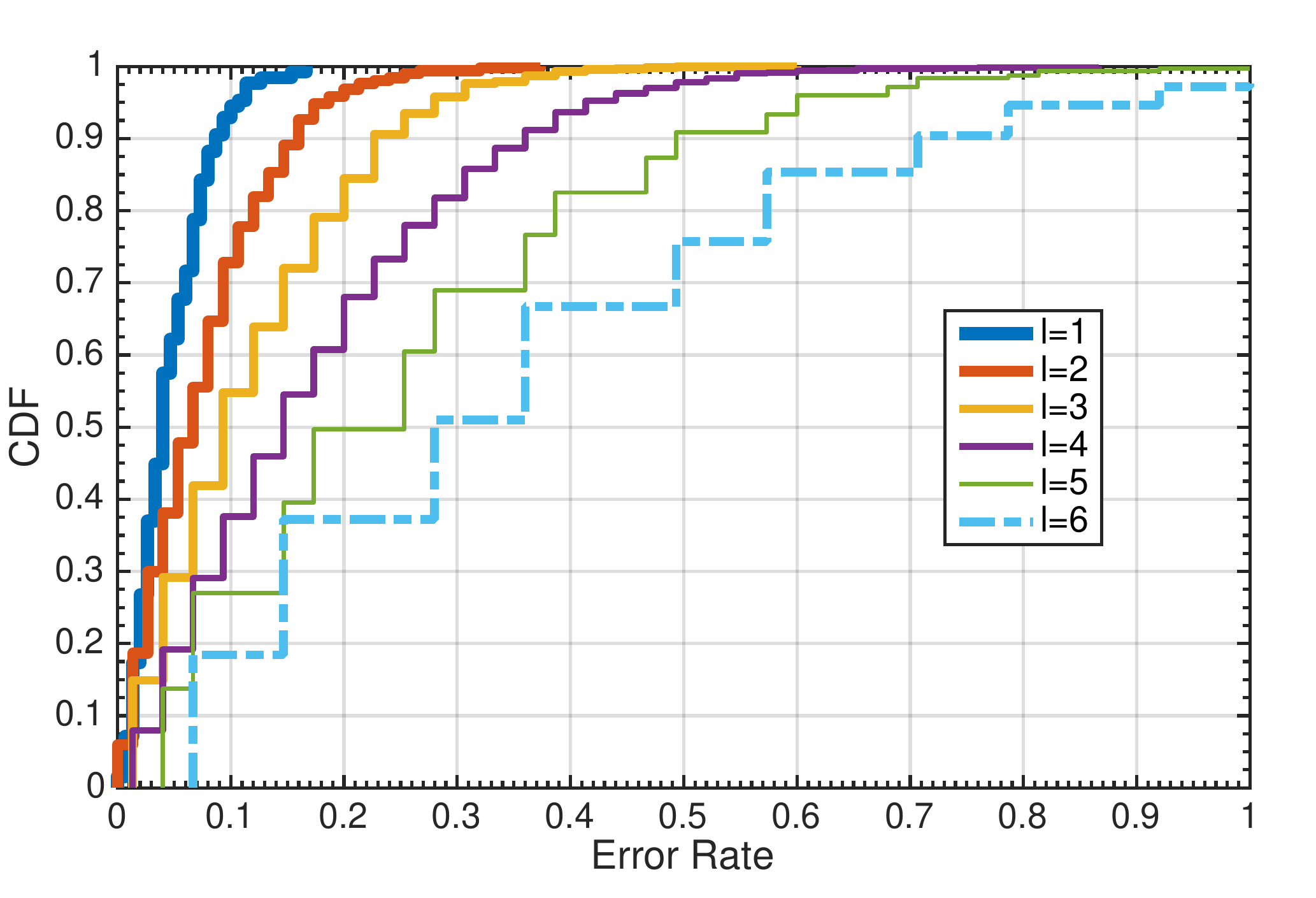}
  }
  \subfigure[Large-scale simulation]{
  	\label{fig:usage-estimation-time}
  	\includegraphics[width=0.8\linewidth]{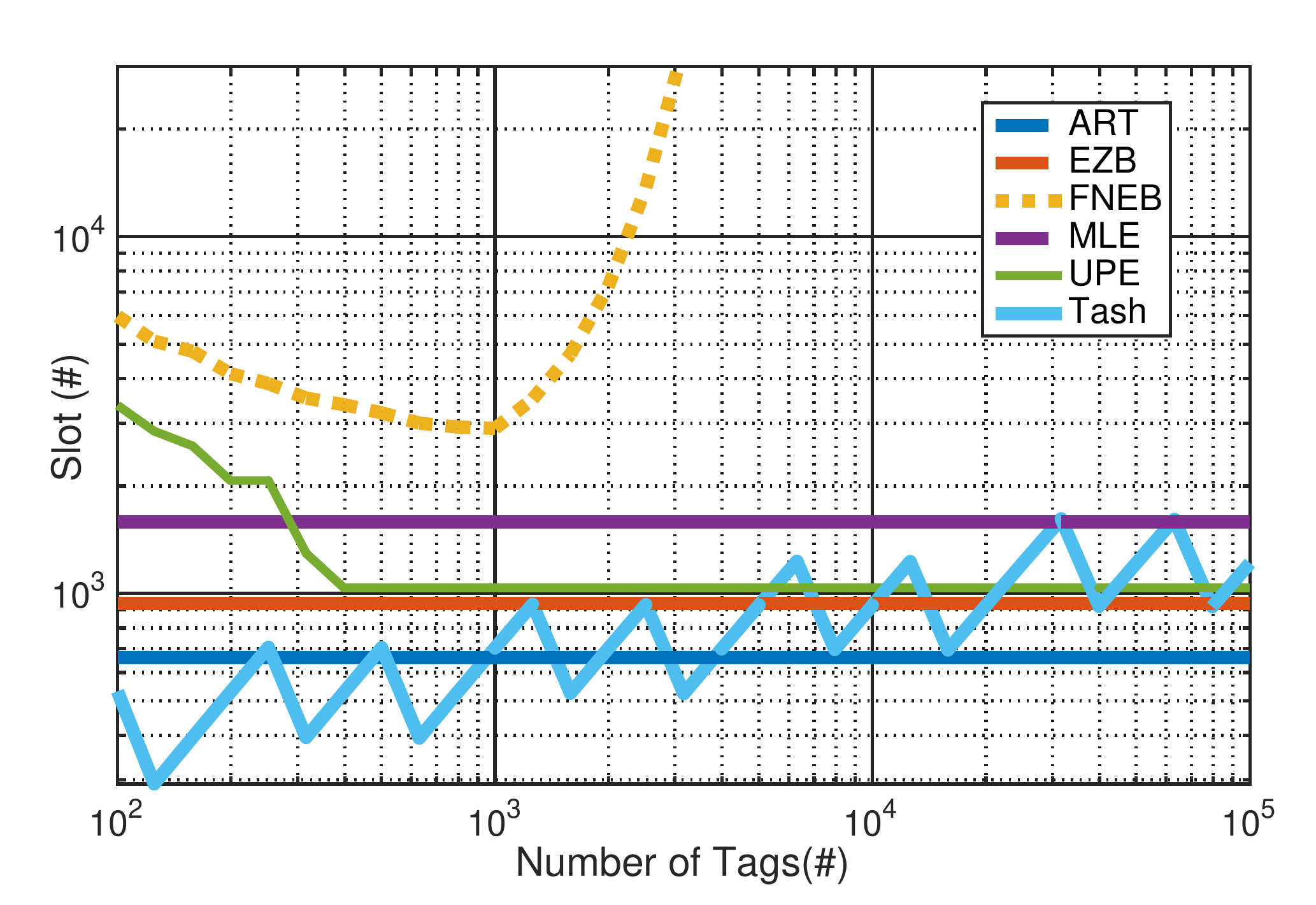}
  }
  \vspace{-0.2cm}
    \caption{Cardinality estimation. \textnormal{(a) shows the CDF of error rates for estimating $300$ tags with our testbed. (b) shows the estimation comparisons with other theoretical algorithms with simulation.}}
    \vspace{-0.3cm}
\end{figure}

Finally, we investigate the performance of tash operators.  Superior to existing HEPs, these operators allow us to perform set operations on-tag and conduct a one-stop inventory. In particular, we show the performance of OR as a representative across $300$ tags. The tests for other operators are similar and omitted due to the space limitation. In the experiments, we fix the two random seeds but change the dimension of tash table. Fig.~\ref{fig:hash-operator-or} shows the results of three cases. In Case $1$, we independently produce $2$ tash tables without truncating a reply and conduct the OR in the application layer.  In Case $2$ and Case $3$, we conduct on-tag OR function as \oursystem provides without and with truncating a reply respectively. Consequently, when the dimension equals $2$, Case 1 takes $6,511ms$ on collecting two tables. On the contrary,  the amount of time taken is reduced to $4,578ms$ (\ie $29.7\%$ drop) if we perform an on-tag OR function even without truncation (Case 2).  Ideally, the amount of time taken could be further reduced to $50.97ms$ by using a truncating reply (Case 3), which offers a staggering drop in time usage by $99.22\%$. Our experiments relate only to the amount of time spent on ORing two tables. It may be predicted that  much more outperformance will be gained if multiple tables are involved. The tash operators that we design in this work have never been proposed before.

\section{Usage Evaluation}
\label{section:usage-evaluation}

We then use our prototype to demonstrate the benefits and potentials of \oursystem in two typical applications.

\subsection{Usage I: Cardinality Estimation}

We evaluate our estimation scheme through the testbed as well as large-scale simulations.

\begin{figure}[t!]
  \centering
  \subfigure[$k=2$ and $l=8$]{
  	\label{fig:missing-detection-eval-1}
  	\includegraphics[width=0.8\linewidth]{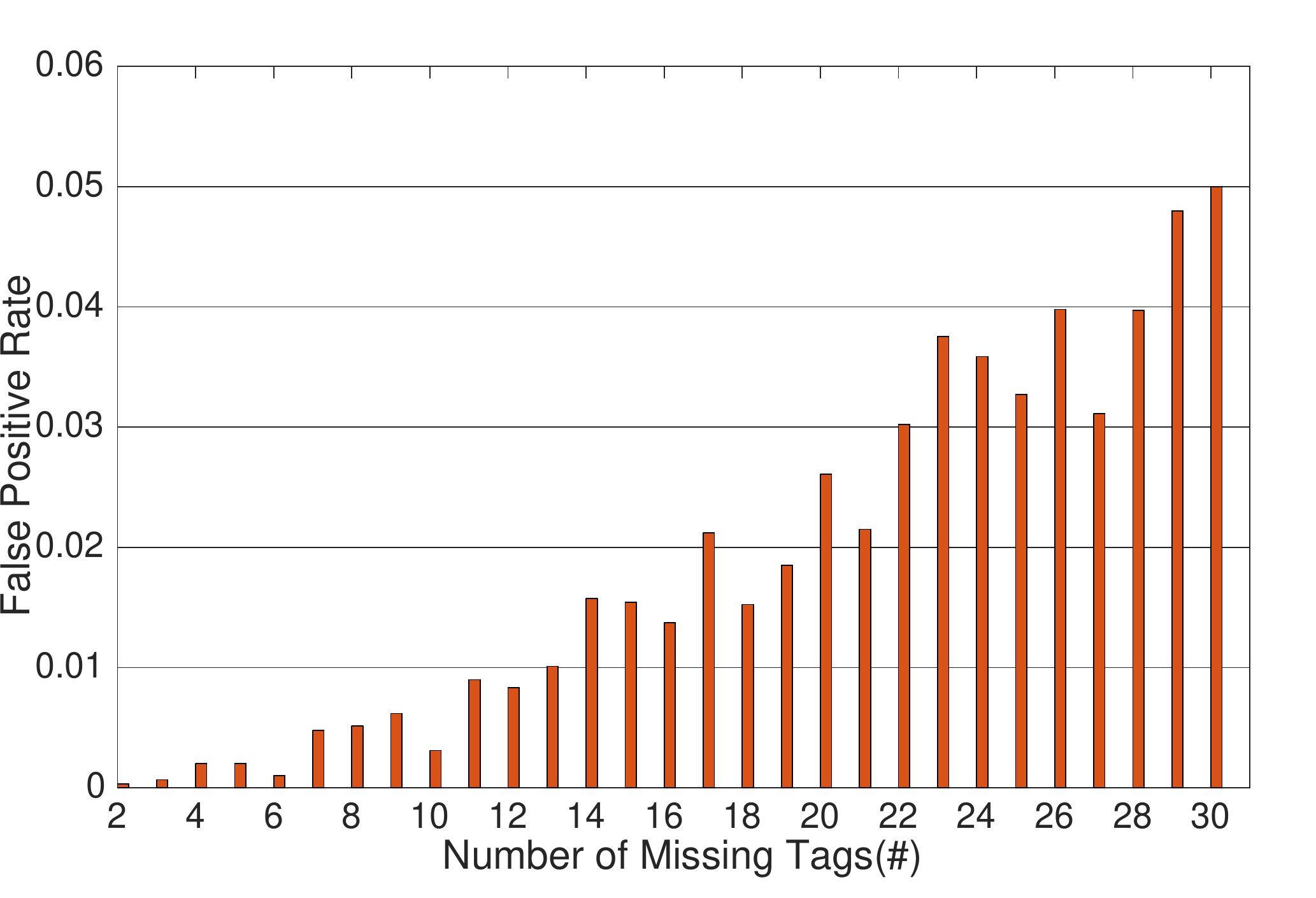}
  }
  \subfigure[$k=2$ and $m=10$]{
  	\label{fig:missing-detection-eval-2}
  	\includegraphics[width=0.8\linewidth]{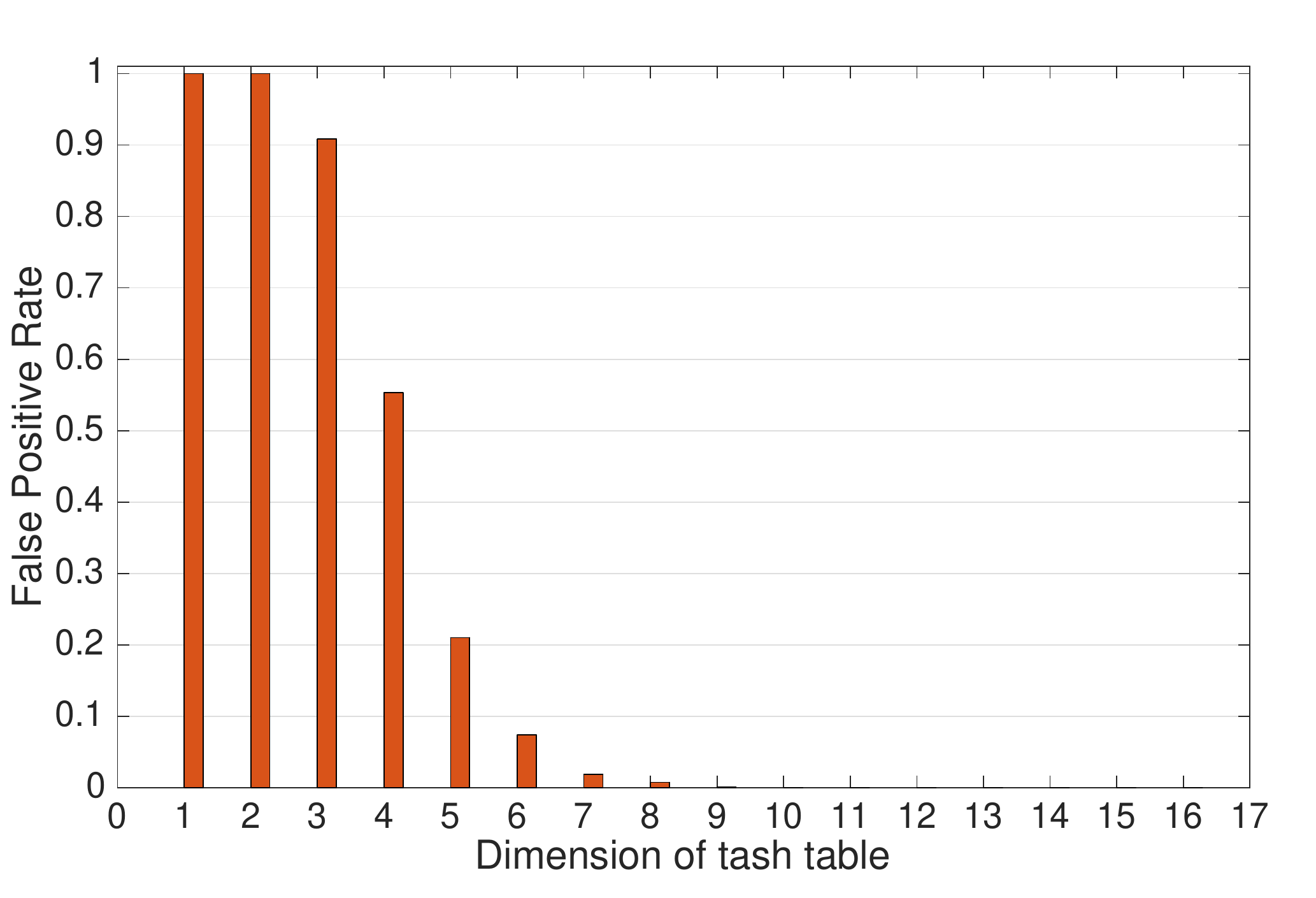}
  }
  \vspace{-0.2cm}
    \caption{Missing detection with $300$ tags. \textnormal{ (a) The resulted FPRs as function of the missing number (b) as function of the dimension of tash tables.}}
       \vspace{-0.3cm}
\end{figure}

\textbf{Testbed based.} Our scheme only uses the first entry of the tash table for the estimation, thereby we only need one entry-inventory. Fig.~\ref{fig:usage-estimation-real} shows the CDF of estimation results across $300$ tags. We define the error rate as $|n-\widehat{n}|/n$ where $\widehat{n}$ is the estimated number.  As a result, $90\%$ of the estimations have an error rate less than $0.1$ and a median of $0.04$ when setting the dimension $l=1$. In this case, almost half tags follow into the first entry so the rate could be pretty high, at the price of longer inventory time. As $l$ increases, the error rate also increases because less samples are acquired for the estimation. These experiments show the feasibility of using tash table for cardinality estimation.

\textbf{Simulation based.} We then perform the evaluation through large-scale simulations for two reasons: ($1$) ensuring its scalability when meeting a huge number of tags. ($2$) making comparisons with prior work, which are all simulation-based. We numerically simulate in Matlab using tash scheme as well as other five prior RFID estimation schemes: UPE\cite{kodialam2006fast}, EZB\cite{kodialam2007anonymous}, FNEB\cite{han2010counting}, MLE\cite{li2010energy}, ART\cite{shahzad2012every}. We implement these schemes by referring to the RFID estimation tool developed by Shahzad\cite{shahzad}. Fig.~\ref{fig:usage-estimation-time} shows the time cost with a varying $n$ given $\alpha=0.9$ and $\beta=0.08$. We observe that our scheme is $5\times$ faster than the others on average when $n<1000$. Thus our scheme is suitable for the estimation with a small number of tags. When $n>1000$, the performance of our scheme starts to vibrate between ART and MLE, due to two reasons. First, our scheme is not collision-free so that more efforts are required to deal with the collisions incurred by more tags. Second,  the size of a tash table can only increase in the power of two, making the size always vibrate around the optimal one.  Even so, the advantage of our scheme is still clear: it is the first RFID estimation scheme that can work in real life. Notice that ART claimed to work with RFID systems because they are theoretically compatible with ALOHA protocols.   Actually, the current COTS RFID systems do not allow user to control the low-level access, like fined-grained adjustment of  frame length and obtaining slot-level feedback, which are necessary to implement ART. Thus, there is no way for ART to implement their algorithms over COTS RFID systems without any hardware  modification and fabrication.

\subsection{Usage II: Missing Detection}

Finally, we evaluate the effectiveness of missing detection in real case. We randomly remove $m$ tags from the testbed. Since we only have $300$ tags in total, we fix the number of random seeds to $2$, \ie $k=2$. The performance is evaluated in term of the false positive rate (FPR), which is the ratio of number of mistakenly detected as missing tags to the total number of really missing tags. Our scheme is able to successfully find out all the missing tags because the residual table always contains the entries that missing tags are tashed into. Fig.~\ref{fig:missing-detection-eval-1} shows the results of the first case in which we use an $8$-bit hash table (\ie $l=8$) to detect the missing tags. Consequently, the FPR is maintained around $0.01$ when $m<14$ (\ie $5\%$ of the tags are missing).  Fig.~\ref{fig:missing-detection-eval-2} shows the second case in which we remove $10$ tags and detect the missing tags by changing the dimension of tash table.  As Theorem.~\ref{theorem:missing-detection} suggests, we should set $l=5,6,7$ to guarantee the FPR $\gamma< 0.2,0.1,0.01$. From the figure, we can find that the results of our experiments completely conform to this theorem. The real FPRs equal $0.21, 0.07$ and $0.008$ in the three cases. Tash enabled missing detection works well in practice.

\section{Conclusion}
\label{section:conclusion}

This work discusses a fundamental issue that how to supplement hash functionality to existing COTS RFID systems, which is dispensable for prior HEPs. A key innovation of this work is our design of hash primitives, which is implemented using selective reading. Tash not only makes a big step forward in boosting prior HEPs, but also opens up a wide range of exciting opportunities.

\section*{Acknowledgments}
The research is  supported by GRF/ECS (NO. 25222917), NSFC  General Program (NO. 61572282) and Hong Kong Polytechnic University (NO. 1-ZVJ3). We thank all the reviewers for their valuable comments and helpful suggestions, and particularly thank Eric Rozner for the shepherd.

\bibliographystyle{ACM-Reference-Format}
\bibliography{acmart.bib} 

\end{document}